\documentclass[11pt]{article}

\usepackage{algorithm}
\usepackage{algpseudocode}

\usepackage{dsfont}
\usepackage{cite}
\usepackage{amsmath}
\usepackage{amsthm}
\usepackage{amssymb}
\usepackage{parskip}
\usepackage{graphicx}
\usepackage{soul}
\usepackage[utf8]{inputenc} 
\usepackage{authblk}
\usepackage{hyperref}       
\usepackage{amsfonts}  
\usepackage{xcolor}
\usepackage{bm}
\usepackage{verbatim}
\usepackage[titletoc,title]{appendix}
\usepackage{tikz}

\usepackage[scale=0.75]{geometry}

\newtheorem{theorem}{Theorem}
\newtheorem{corollary}{Corollary}
\newtheorem{lemma}{Lemma}
\newtheorem{mydef}{Definition}

\def\>{\rangle}
\def\<{\langle}
\def\polylog{{\rm polylog}}
\def\poly{{\rm poly}}
\def\O{\mathcal{O}}
\def\R{\mathbb{R}}
\def\C{\boldsymbol{C}}
\def\b{\boldsymbol{b}}
\def\D{\Delta}
\def\y{\boldsymbol{y}}

\begin{document}

\title{Quantum algorithm for finding the negative curvature direction in non-convex optimization}

\author[1]{Kaining Zhang \thanks{kzha3670@uni.sydney.edu.au}}
\author[2]{Min-Hsiu Hsieh \thanks{Min-Hsiu.Hsieh@uts.edu.au}}
\author[1]{Liu Liu}
\author[1]{Dacheng Tao}
\affil[1]{UBTECH Sydney AI Centre and the School of Computer Science,
Faculty of Engineering and Information Technologies, The University of Sydney, Australia
\authorcr kzha3670@uni.sydney.edu.au}
\affil[2]{Centre for Quantum Software and Information, Faculty of Engineering and Information Technology, University of Technology Sydney, Australia}

\renewcommand\Authands{ and }

\maketitle
\begin{abstract}
We present an efficient quantum algorithm aiming to find the negative curvature direction for escaping the saddle point, which is the critical subroutine for many second-order non-convex optimization algorithms. We prove that our algorithm could produce the target state corresponding to the negative curvature direction with query complexity \(\tilde{\O}(\polylog(d) \epsilon^{-1})\), where \(d\) is the dimension of the optimization function. The quantum negative curvature finding algorithm is exponentially faster than any known classical method which takes time at least \(\O(d\epsilon^{-1/2})\). Moreover, we propose an efficient quantum algorithm to achieve the classical read-out of the target state. Our classical read-out algorithm runs exponentially faster on the degree of $d$ than existing counterparts.
\end{abstract}

\section{Introduction}
\label{QNCD_intro}

Algorithms for finding the minima of functions have attracted significant attention due in part to their prevalent applications in machine learning, deep learning and robust statistics; in particular, those with good complexity guarantees that can converge to the local minima. Numerous algorithms have been proposed in recent years for finding points that satisfying 
\begin{align*}
	\left\| {\nabla f\left( \bm{x} \right)} \right\| \le {\epsilon _g}, \text{and} \ {\lambda _{\min }}\left( {{\nabla ^2}f\left( \bm{x} \right)} \right) \ge {-\epsilon _H},
\end{align*}
where $\epsilon_g,\epsilon_H \in (0,1)$.
A recent proposals \cite{nesterov2006cubic,conn2000trust,agarwal2017finding} based on  second-order Newton-type  and first-order methodology have been analyzed from such a perspective. 
However, those methods normally deal with the situations that the iterations may be trapped in the saddle points, since in many cases, such as deep neural networks \cite{dauphin2014identifying, choromanska2015loss}, existence of many saddle points is the main bottleneck.

In general non-convex optimization, there are many proposed algorithms for escaping the saddle points. These algorithms can be divided into the following two categories: the first-order gradient-based algorithms and the second-order Hessian-based algorithms. Generally, second-order algorithms have better iteration complexity than first-order algorithms (see \cite{jin2017escape} for detail). 
However, each iteration in the second-order method involves the computation of the \textbf{negative curvature direction}, namely, the eigenvectors of a Hessian matrix $\bm{H}=\nabla ^2 f(\bm{x})$ with negative eigenvalues. This computation could take time $\O(d^2)$ when SVD is performed on the given Hessian, or $\O(d/\sqrt{\epsilon})$ when Lanczos method is used with Gradient information to approximate the Hessian-vector product.

Quantum algorithms have shown great potential to become faster alternatives than classical algorithms for many kinds of problems in the field of linear algebra, including principal component analysis \cite{lloyd2014quantum}, support-vector machine \cite{rebentrost2014quantum}, singular value decomposition \cite{rebentrost2018quantum}, etc.. 
These works encourage us to develop an efficient quantum algorithm for the Negative Curvature Finding problem, which aims to outperform the best known classical methods.
To begin with, we formally define the negative curvature finding problem as follows.

\noindent\textbf{Negative Curvature Finding (NCF) problem:}
\textit{
Given the function \(f(\bm{x}):\mathbb{R}^{d} \rightarrow \mathbb{R}\) which has \(L\)-Lipschitz continuous gradient, and the corresponding Hessian matrix \(\bm{H} \in \mathbb{R}^{d \times d}\), we aim to build a quantum algorithm that could efficiently provide the unit vector \(\bm{u}\) with the condition: }
\begin{equation}\label{vector_property}
\bm{u}^{T} \bm{H} \bm{u} \leq -\alpha+\epsilon,
\end{equation}
\textit{where \(0<\alpha < L\) and \(0<\epsilon<\alpha\); or make the non-vector statement that with high probability there is no unit vector $\bm{u}$ satisfying the following condition:}
\begin{equation}\label{non-vector}
\bm{u}^{T} \bm{H} \bm{u} < -\alpha.	
\end{equation}

\subsection{Related work}

Optimization methods for non-convex problems can be roughly divided into first-order methods and second-order methods, depending on the order of the derivative to the objective function they used. Generally, in order to find the local minima,  the second-order methods \cite{carmon2018accelerated,agarwal2017finding} are applied to find the effective direction to escape the saddle point. Specifically, finding the Negative Curvature is considered as the subroutine to analyze the characteristic of the saddle point.

\textbf{First-order algorithms}: 
For the non-convex problem, the first-order method (Gradient-based method) can find the stationary point, which could be a global minima, local minima or saddle point. Finding the global minima is an NP-hard problem, and many methods instead are trying to find the local minimum. However, standard analysis of gradient descent cannot distinguish between saddle points and local minima, leaving open the possibility that gradient descent may get stuck at saddle points. Recently Ge \textit{et al.} \cite{ge2015escaping,jin2017escape,jin2019stochastic}
showed that by adding noise at each step, gradient descent can escape all saddle points in a polynomial number of iterations, provided that the objective function satisfies the strict saddle property \cite{ge2015escaping}.
Lee \textit{et al.} \cite{lee2016gradient} proved that under similar conditions, gradient descent with random initialization avoids saddle points even without adding noise. However, each iteration of Gradient-based methods requires $\O(d)$ operations and the iteration complexity is higher than second-order algorithms \cite{jin2017escape}.

\textbf{Second-order algorithms}: Traditionally, second-order Newton-based methods can converge to local minima, which use the Hessian information to distinguish between first-order and second-order stationary points. There are two kinds of methods that make use of Hessian information. 1) Hessian-based: trust-region \cite{conn2000trust} and cubic regularization \cite{nesterov2006cubic} are two methods, in which the sub-problem is to find the decrease direction based on the given Hessian matrix. The calculation of each iteration involves performing SVD on the Hessian matrix, which takes time at least $\O(d^2)$. 2) Hessian-vector-product-based: 
While the subproblem that appears in the cubic-regularized Newton method is expensive to solve exactly, it is possible to consider methods in which such subproblems are solved only approximately by Hessian-free procedures. 
The Hessian-vector-product method uses Lanczos method to calculate the negative curvature direction and uses gradient to approximate the Hessian-vector product \cite{agarwal2016finding, carmon2018accelerated, carmon2016gradient}. The Hessian-vector-product method involves $\O(d\epsilon^{-1/2})$ complexity per iteration.
The advantage of the second-order algorithm is the superior iteration complexity than the first-order algorithm. However, using Hessian information usually increases computation time per iteration.

On the other hand, there are also some proposed quantum algorithms for problems in the related linear algebra field. 
For example, previous quantum PCA algorithm \cite{lloyd2014quantum} presents an efficient way to do the Hamiltonian simulation task in quantum phase estimation. The time complexity to perform the mapping $\sum_j \beta_j |\bm{u}_j\> \rightarrow \sum_j \beta_j |\bm{u}_j\> |\tilde{\lambda}_j\>$ could be bounded in \(\O(\polylog(d)\epsilon^{-3})\). However, the quantum PCA model use the density matrix $\rho=\sum_{i,j=1}^d a_{ij} |i\>\<j|$ to store the information of matrix $X=A^{\dag}A$, which implicitly assumes the condition $\|X\|_F=1$ and $\lambda_{\min}(X)\footnote{$\lambda_{\min}(X)$ means the least eigenvalue of matrix $X$.} \ge 0$.
Another quantum SVD algorithm \cite{rebentrost2018quantum} shows an efficient method to estimate the value \(\lambda_j /d\) with error \(\epsilon\) in time \(\O(\epsilon^{-3})\), where \(\{\lambda_j\}_{j=1}^d\) are eigenvalues of matrix \(\bm{H} \in \R^{d \times d}\)  . However, it would take time \(\O(\poly(d))\) to produce \(\epsilon\)-estimation on eigenvalues by this quantum SVD algorithm. Moreover, these works did not study the classical read-out of the quantum output, which actually takes time at least \(\O(d)\) generally for \(d\)-dimensional state \cite{aaronson2015read}, and could offset the claimed quantum speed-up.

\subsection{Our contribution}

The contribution of this work can be briefly divided into two parts: 1) an efficient quantum algorithm to generate the required quantum state, which corresponds to the negative curvature direction, and 2) an efficient quantum algorithm to obtain the description of the target state $|\bm{u}_t\>=\sum_{i=1}^{r}x_i |\bm{s}_i\>$, where $\{\bm{s}_i\}_{i=1}^{r}$ is an independent vector set selected from columns of Hessian $\bm{H}$ with rank $r$.

\textbf{Negative Curvature Finding}: We develop an efficient quantum algorithm to produce the target state $|\bm{u}_t\>$ (for case (\ref{vector_property})) or make the non-vector statement (for case (\ref{non-vector})). The algorithm contains three subroutines:
\begin{itemize}
\item A quantum algorithm to distinguish different eigenvectors through that the corresponding eigenvalue is positive or negative. The algorithm takes the eigenstate as the input and outputs the statement that the corresponding eigenvalue is positive or negative. Previous quantum SVD \cite{rebentrost2018quantum} could only distinguish different eigenstates through different absolute value of eigenvalues, so we provide this subroutine to deal with this problem.
\end{itemize}
\begin{itemize}
\item A quantum algorithm to label the proper eigenvalue (less than $-\alpha+\epsilon/2$, for case (\ref{vector_property})), or to make the non-vector statement (for case (\ref{non-vector})).
\end{itemize}
\begin{itemize}
\item A quantum algorithm to generate the target vector $\bm{u}_t$ in the quantum state form $|\bm{u}_t\>$. 
\end{itemize}
We provide Theorem \ref{QNCF_Theorem_1} as the main result of this part, which guarantees the time complexity of our Negative Curvature Finding algorithm:
\begin{theorem}
\label{QNCF_Theorem_1}
There exists a quantum algorithm which could solve the Negative Curvature Finding problem in time $\tilde{\O}(\polylog(d) \poly(r)\epsilon^{-1})$, by providing the target state $|\bm{u}_t\>$ (for case (\ref{vector_property})), or making the non-vector statement (for case (\ref{non-vector})).
\end{theorem}

\textbf{Classical Read-out}: 
The classical read-out problem is one bottleneck for many quantum machine learning algorithms whose results are quantum states. Generally, the read-out of a $d$-dimensional quantum state takes time at least \(\O(d)\) \cite{aaronson2015read}, and could offset the claimed quantum speed-up.
In order to solve this dilemma, we develop an efficient quantum algorithm for the classical read-out of the target state. We notice that the target vector \(\bm{u}_t\) can be written as the linear combination of \(\{\bm{s}_i\}_{i=1}^{r}\), where \(\{\bm{s}_i\}_{i=1}^{r}\) is an independent vector group sampled from column vectors \(\{\bm{h}_j\}_{j=1}^{d}\). Thus one could obtain the coordinate $|\bm{u}_t\>=\sum_{i=1}^{r} x_i |\bm{s}_i\>$ by solving the $r$-dimensional linear system $\bm{C}\bm{x}=\bm{b}$, in which $\bm{C}=\{c_{ij}\}_{ij}=\{\<\bm{s}_i|\bm{s}_j\>\}_{ij}$ and $b_j=\<\bm{s}_j|\bm{u}_t\>$. The algorithm suits the case when the result quantum state lies in the span of several given states, and may give rise to independent interest.

One advantage of generating the form $|\bm{u}_t\>=\sum_{i=1}^{r} x_i |\bm{s}_i\>$ is that the updating operation for the $l$-th iteration $\bm{z}^{(l+1)}= \bm{z}^{(l)} +\eta \bm{u}_t$ in general non-convex algorithms could be implemented in quantum form $|\bm{z}^{(l+1)}\> \propto |\bm{z}^{(l)}\> +\eta \sum_{i=1}^{r} x_i |\bm{s}_i\>$ by Linear-Sum-of-States method \cite{shao2018linear} efficiently, which may inspire efficient quantum counterparts for these non-convex algorithms. Note that we could not perform the operation $|\bm{z}^{(l+1)}\> \propto |\bm{z}^{(l)}\> +\eta |\bm{u}_t\>$ directly since the target state $|\bm{u}_t\>$ is generated by post-selection instead of standard unitary operations. Our state read-out algorithm contains two subroutines named as the Complete Basis Selection and the State Overlap Estimation. The main results about the Complete Basis Selection and the Classical Read-out are briefly summarized as following Theorems:

\begin{theorem}\label{Basis_theorem_1}
There exists a quantum algorithm which takes time $\tilde{\O}(\poly(r) \epsilon^{-2} r^c)$ to find an index set $\{g(i)\}_{i=1}^{r}$, where $r$ is the rank of $\bm{H}$, $c=2\log\frac{4r\|\bm{H}\|_F}{\epsilon}$ and $\{g(i)\}_{i=1}^{r}$ forms a complete basis $\{|\bm{h}_{g(i)}\>\}_{i=1}^{r}$ with probability at least 3/4.	
\end{theorem}

\begin{theorem}
\label{solution_1}
The classical description of the target state ${|\bm{u}}_t\>=\sum_{i=1}^{r} {x}_i |\bm{s}_i\>$ could be presented in time \(\tilde{\O}(\polylog(d) \poly(r) \epsilon^{-5})\) with error bounds in \(\epsilon/2\), when the complete basis set $\{\bm{s}_j\}_{j=1}^{r}$ is given. 
$\C$ is the $r\times r$ Gram matrix defined as $\bm{C}=\{c_{ij}\}_{ij}=\{\<\bm{s}_i|\bm{s}_j\>\}_{ij}$.
\end{theorem}

The rest of this paper is organized as follows. Some preliminaries about quantum information and other useful notations and definitions are introduced in Section \ref{QNCD_preliminary}. 
In Section \ref{QNCD_state_finding}, we develop an quantum algorithm to solve the \textbf{NCF problem}. 
In Section \ref{QNCD_state_read-out}, we develop an quantum algorithm which aims to read out the target state. We summarize our results and contributions in Section~\ref{QNCD_conclu}.

\section{Preliminary}
\label{QNCD_preliminary}

In this section we present some preliminary concepts, which play vitally important roles throughout this paper. 
Some basic quantum knowledge along with useful notations and definitions will be introduced in Section \ref{QNCD_pre_note}. 
Some quantum technics such as quantum oracle models and quantum singular value estimation algorithm will be introduced in Section \ref{QNCD_pre_technic}.

\subsection{Notations and definitions}
\label{QNCD_pre_note}

In this section, we introduce some useful notations and definitions. Since the quantum notations are critically important in the following sections, we would briefly introduce some basic quantum information knowledge first. Then we introduce some other useful notations and definitions.

The dirac notation is a standard notation in quantum mechanics to describe the quantum states. 
The form \(|\bm{x}\>\) is the state which corresponds to the vector \(\bm{x}\), and the form \(\<\bm{y}|\) is the state which corresponds to the vector \(\bm{y}^{T}\). 
The notation \(\< \bm{y} | \bm{x} \>\) denotes the inner product \(\bm{y}^{T} \bm{x}\). The notation \(|\bm{y} \> \< \bm{x}|\) denotes the matrix \(\bm{y} \bm{x}^{T}\). 
Quantum state is unitary, which means \(\||\bm{x} \>\|^2 = \< \bm{x} | \bm{x} \> = 1\). 
Thus for vector \(\bm{x} \in \R^d\), the state \(|\bm{x}\>\) is defined as \(\frac{1}{\|\bm{x}\|} \sum_{j=1}^d x_j |j\>\), where \(x_j\) is the \(j\)-th component of vector \(\bm{x}\) and \(\{|j\>\}_{j=1}^{d}\) is the state basis which acts like \(\{\bm{e}_j\}_{j=1}^d\) in classical case. 
One significant difference between the classical vector \(\bm{x}\) and the quantum state \(|\bm{x}\>\) is that we could not get the detail of \(x_j\) with \(\O(1)\) queries to state \(|\bm{x}\>\). 
The only way to generate classical information from \(|x\>\) is by measurement. The measurement operation could be viewed as the biased coin experiment. 
For example, considering the state \(|x\> = \sum_{j=1}^{d} x_j/\|\bm{x}\||{j}\>\), the measurement of \(|x\>\) on the basis  \(\{|j\>\}_{j=1}^{d}\) could randomly produce different index \(j\) with probability \(x_j^2/\|\bm{x}\|^2\).

We use \([n]\) to denote the set \(\{ 1,2,\cdots,n \}\).
We denote the norm \(\|\cdot\|\) as the \(\|\cdot\|_2\) norm for vectors, if there is no more explanation. \(\|A\|_F = (\sum_{i=1}^{m}\sum_{j=1}^{n} a_{ij}^2)^{1/2}\) is the Frobenius norm of matrix \(A \in \R^{m \times n}\). The lowercase form \(\bm{h}_i\) is defined as the \(i\)-th column vector of matrix \(\bm{H} \in \R^{d \times d}\). \(x_{i}\) is defined as the \(i\)-th component of vector \(\bm{x}\).
The tensor product of two matrix \(\bm{A} \in \R^{m \times n}\) and \(\bm{B} \in \R^{p \times q}\) is defined as \(\bm{C}=\bm{A} \otimes \bm{B}\).
The tensor product operation could be performed between vectors, since vector is one special kind of matrix. 
The tensor product could be defined between quantum states $|\bm{x}_1\>$ and $|\bm{x}_2\>$, for example, \(|\bm{x}\> =|\bm{x}_1\> \otimes |\bm{x}_2\>\). 
The form \(|\bm{x}_1\> \otimes |\bm{x}_2\>\) could also be written as \(|\bm{x}_1\> |\bm{x}_2\>\). 

We present definitions of smoothness and $\gamma$-separation here.
\begin{mydef}{\rm (smoothness)}
A function \(f: \R^d \rightarrow \R\) is \(L\)-smooth if it has \(L\)-Lipschitz continuous gradient, that is \(\|\nabla f(\bm{x}) - \nabla f(\bm{y})\| \leq L \|\bm{x}-\bm{y}\|\), \(\forall \bm{x},\bm{y} \in \mathcal{X}\), where \(\mathcal{X}\) is the domain of \(f(\bm{x})\).
\end{mydef}

\begin{mydef}{\rm (\(\gamma\)-separation)}
The set \(G=\{a_1,a_2,\cdots, a_n\}\) is said to be \(\gamma\)-separated if \(|a_i-a_j|> \gamma, \forall i,j \in [n] \) and \(i \neq j\).
\end{mydef}

Based on these definitions, we assume that the Hessian matrix $\bm{H}$ in this article has two properties:


\begin{enumerate}
\label{property}
\item  $\bm{H} \in \R^{d \times d}$ is a $r$-rank Hessian matrix which is derived from the \(d\)-dimensional optimization problem \(\min_{\bm{x} \in {\R}^{d}} f(\bm{x})\) in which the objective function \(f\) has \(L\)-Lipschitz continuous gradient;
\item  The absolute value of $\bm{H}$'s non-zero eigenvalue is \(\epsilon\)-separated.
\end{enumerate}

The first property is directly derived from the assumption of previous classical non-convex optimization method \cite{carmon2018accelerated}, and the low-rank Hessian case has been observed in neural networks\cite{gur2018gradient}.
The second property is assumed such that we could distinguish different eigenvalues by their absolute value. 
We further assume that the Hessian matrix $\bm{H}$ has the eigen-decomposition $\bm{H}=\sum_{j=1}^{r} \lambda_j \bm{u}_j \bm{u}_j^{T}$, for the convenience of following discussion.

\subsection{Techniques}
\label{QNCD_pre_technic}
The motivation idea behind our approach is to perform the quantum singular value estimation model and then generate eigen-states by the post-selection on the output state. 
Here we introduce some techniques including oracle models and critical conclusions in previous work. 

\subsubsection{Quantum Oracle Models\cite{kerenidis2016quantum}}
\label{QNCD_pre_technic_oracle}

For the whole paper, we assume the existence of following quantum oracles, and discuss the query complexity of our algorithms to these oracles.
Given Hessian \(\bm{H} \in \R^{d \times d}\), we assume that $\bm{H}$ is stored in a classical data structure such that the following quantum oracles could be implemented:
\begin{align}
\label{QNCD: equation:oracle_1}
U_H &:|i\> |0\> \stackrel{}\rightarrow |i\> |\bm{h}_i\> =\frac{1}{\|\bm{h}_i\|} \sum_{j=1}^{d} h_{ij} |i\> |j\>, \forall i \in [d], \\
\label{QNCD: equation:oracle_2}
V_H &: |0\> |j\> \stackrel{}\rightarrow |\tilde{\bm{h}}\> |j\> = \frac{1}{\|\bm{H}\|_F} \sum_{i=1}^{d} \|\bm{h}_i\| |i\> |j\>, \forall j \in [d],
\end{align}
where \(\tilde{\bm{h}}\) stands for the \(d\)-dimensional vector whose \(i\)-th component is \(\|\bm{h}_i\|/\|\bm{H}\|_F\).

The required data structure has a binary tree form. The sign and square value for each entry are stored in different leaves and the value stored in each parent node is the sum of its children's value. A detail description about this data structure can be referred to \cite{kerenidis2016quantum}. Denote $T_H$ as the time complexity of these oracles.

\subsubsection{Quantum Singular Value Estimation (SVE)}

Given matrix $\bm{A} \in \R^{m \times n}$ which has the singular value decomposition $\bm{A}=\sum_{j=1}^{\min(m,n)} \sigma_j \bm{u}_j \bm{v}_j$, 
previous work \cite{kerenidis2016quantum} provided a quantum singular value estimation algorithm, which could be used for estimating singular value or generating eigenstate. Here we briefly introduce their conclusion about the time complexity of their algorithm:

\begin{theorem}\cite{kerenidis2016quantum}
\label{quantum_SVD}
Suppose matrix \(\bm{A} \in \R^{m \times n}\) is stored in the data structure in Section \ref{QNCD_pre_technic_oracle}. Let $\epsilon$ be the precision parameter. There is an algorithm which could perform the mapping \(\sum_j \beta_j |\bm{v}_j\> \rightarrow \sum_j \beta_j |\bm{v}_j\> |\hat{\sigma}_j \>\) with query complexity \(\O(\polylog(n)\epsilon^{-1})\), where \(\hat{\sigma}_j \in \sigma_j\pm \epsilon \|\bm{A}\|_F\) with probability at least \(1-1/\poly(n)\).
\end{theorem}

\subsubsection{Linear Sum of States}

The idea of linear combination of states was introduced in \cite{shao2018linear}, which focuses on the following problem: given quantum states $|a\>$ and $|b\>$, to prepare the state $|c\>=\frac{1}{Z_c}(x|a\>+y|b\>)$. The method is based on the fact that $R_{2\theta}=(I-2|b\>\<b|)(I-2|a\>\<a|)$ can be viewed as the clockwise rotation in the plane spanned by $|a\>$ and $|b\>$ with angle $2\theta$, where $\theta=\arccos\<a|b\>$ is the angle between $|a\>$ and $|b\>$. Thus any clockwise rotation in space $SPAN\{|a\>,|b\>\}$ with angle $\phi$ could be written as $R_{\phi}=R_{2\theta}^{\phi/2\theta}$. For the case $|c\>=\frac{1}{Z_c}(x|a\>+y|b\>)$, there is $|c\>=R_{\phi}|a\>$, where $\phi=\arccos\frac{x+y\<a|b\>}{Z_c}=\arccos\frac{x+y\<a|b\>}{\sqrt{x^2 +y^2 +2xy\<a|b\>}}$.
The linear sum of two states could be generalized to $n$ case: 

\begin{theorem}{\rm \cite{shao2018linear}}
Assume state $|\phi_i\>$ could be prepared by given unitary operation in time $T_{in}$, for $i \in [n]$. Then there is a unitary which could prepare the state $|\phi\>=\sum_{i=1}^n \alpha_i |\phi_i\>$ in time $\O(T_{in} n^{\log(n/\epsilon)})$ with error $\epsilon$.
\end{theorem}

\section{Quantum Negative Curvature Finding algorithm}
\label{QNCD_state_finding}

Our main contribution in this section is the quantum Negative Curvature Finding (quantum NCF) algorithm presented in Algorithm \ref{QNCF_A}. The quantum NCF algorithm solves the NCF problem by providing the target state $|\bm{u}_t\>$ (for case (\ref{vector_property})) or making the non-vector statement (for case (\ref{non-vector})). The target state \(|\bm{u}_t\>\) corresponds to the eigenvector \(\bm{u}_t\) which satisfies the condition $\bm{u}_t^{T}\bm{H}\bm{u}_t \le -\alpha +\epsilon/2$. 
Here we present a tighter restrict on the target state $|\bm{u}_t\>$ to keep a $\epsilon/2$ redundancy for the classical read-out of the quantum state.
The quantum NCF Algorithm uses the Proper Eigenvalue Labelling (Algorithm \ref{Label_proper}) and the Target State Generating (Algorithm \ref{target_state}) as subroutines proposed in Section \ref{3_posit_negat} and Section \ref{3_input_state}, respectively.

\begin{algorithm}[t]
  \caption{Quantum Negative Curvature Finding (Quantum NCF) Algorithm}
  \label{QNCF_A}
  \begin{algorithmic}[1]
    \Require
      The Hessian matrix $\bm{H}$ which is stored in the data structure in Section~\ref{QNCD_pre_technic_oracle}. The parameter \(\epsilon\) and \(\alpha\) in the \textbf{NCF problem}.
    \Ensure
      The target state $|\bm{u}_t\>$ whose corrsponding classical unit vector $\bm{u}_t$ satisfies the condition $\bm{u}_t^T \bm{H} \bm{u}_t \le -\alpha + \epsilon/2$; or a statement with high probability that there is no such kind of unit vector $\bm{u}$ which satisfies the condition $\bm{u}^T \bm{H} \bm{u} \le -\alpha $.
    \State Label the \textbf{proper}(less than $-\alpha+\epsilon/2$) eigenvalue of $\bm{H}$ (Algorithm \ref{Label_proper}). 
    \label{Qcurvature_negaeigen_init}
    \If {the least eigenvalue of $\bm{H}$ is less than $-\alpha+\epsilon/2$,}
    \State generate the target state (Algorithm \ref{target_state}) and output the state;
    \label{Qcurvature_negaeigen_hard}
    \Else ,
    \State claim that there is no such kind of unit vector $\bm{u}$ which satisfies the condition $\bm{u}^T \bm{H} \bm{u} \le -\alpha $.
    \EndIf
  \end{algorithmic}
\end{algorithm}

\subsection{Challenges to Develop Quantum NCF algorithm}

The core technical component of our quantum algorithm for the NCF problem is the quantum SVE algorithm. However, there are three major challenges that we have to overcome. 

Firstly, the positive-negative eigenvalue problem. In the negative curvature finding problem, we are interested in obtaining eigenvectors with negative eigenvalues. Hence, we can not directly apply the quantum SVE algorithm since it only gives the estimation on \(|\lambda_j|\). In order to overcome this critical issue, we develop Algorithm~\ref{negative_eigenvalue} to label negative eigenvalues. 

Secondly, since the quantum SVE Algorithm presents $\epsilon$-estimation on singular values with time complexity $\O(T_H\|\bm{H}\|_F \polylog(d)\epsilon^{-1})$(Theorem \ref{quantum_SVD}), we need to provide a tight upper bound for the Frobenius norm $\|\bm{H}\|_F$, which is shown in Lemma \ref{bounded_norm} (proof is in Appendix A):

\begin{lemma}
\label{bounded_norm}
Suppose \(\bm{H} \in \R^{d \times d}\) is the Hessian matrix derived from the function \(f : \R^{d} \rightarrow \R\) which has the \(L\)-Lipschitz continuous gradient. Thus the Frobenius norm of $\bm{H}$ has the upper bound \(\|\bm{H}\|_F \leq \sqrt{r} L\), where \(r\) is the rank of $\bm{H}$.
\end{lemma}

Finally, the input-state problem. For the general superposition state \(\sum_{j} \beta_j |\bm{u}_j\>\), the output state of quantum SVE algorithm has the form \(\sum_j  \beta_j |\bm{u}_j\> ||\hat{\lambda}_j |\>\). 
We could generate different pure state $|\bm{u}_j\>$ with probability $|\beta_j|^2$ by the measurement on eigenvalue register. 
Thus in order to guarantee a small time complexity, we need to prepare a special input state such that the overlap between the input and the target state is relatively large. 

We briefly summarize our conclusion on the time complexity of Algorithm \ref{QNCF_A} in Theorem \ref{QNCF_Theorem}.

\begin{theorem}
\label{QNCF_Theorem}
Algorithm \ref{QNCF_A} takes time $\O(T_H \|\bm{H}\|_F^5 \polylog(d)\epsilon^{-1})$ to solve the negative curvature finding problem by providing the target state $|\bm{u}_t\>$ or making the statement that there is no unit vector satisfies the condition $\bm{u}^T \bm{H} \bm{u} \le -\alpha $. 
\end{theorem}
\begin{proof}
The time complexity of Algorithm \ref{QNCF_A} could be directly obtained by the time complexity of Algorithm \ref{Label_proper} and Algorithm \ref{target_state}, whose complexity analysis are presented in Theorem \ref{QNCF_claim_theorem} and Theorem \ref{state_finding}, respectively. 
\end{proof}

\subsection{Positive-Negative Eigenvalue Discrimination}
\label{3_posit_negat}

In this section, we propose an algorithm aiming to label the target eigenvalue which is less than $-\alpha+\epsilon/2$. This algorithm helps verifying the existence of solution to the NCF problem and generating the target state. 
Since we only have the estimating on singular values by quantum SVE algorithm, we need to first develop Algorithm \ref{negative_eigenvalue} which helps to make the statement that the corresponding eigenvalue is positive or negative. 

  \begin{algorithm}[t]
  \caption{Positive-Negative Eigenvalue Discrimination(PNED) Algorithm}
  \label{negative_eigenvalue}
  \begin{algorithmic}[1]
    \Require
      Quantum oracles $U_H$ and $V_H$. The eigenstate \(|\bm{u}\>\) whose corresponding eigenvalue is $a$.
    \Ensure 
    A measurement result which has different values $0$ and $1$ with probability $P(0)=\frac{1+\lambda/\|\bm{H}\|_F}{2}$ and $P(1)=\frac{1-\lambda/\|\bm{H}\|_F}{2}$.
    \State Create state \(|\bm{u}\>|0\>|0\>\). The second register has the same qubit length with state \(|\bm{u}\>\) and the third register has one qubit length.
    \label{Qcurvature_negaeigen_init}
    \State Apply the Hadmard gate on the third register to obtain the state \(\frac{1}{\sqrt{2}}(|\bm{u}\>|0\>|0\>+|\bm{u}\>|0\>|1\>)\).
    \label{Qcurvature_negaeigen_hard}
    \State Apply the controlled SWAP gate to obtain the state \(\frac{1}{\sqrt{2}}(|\bm{u}\>|0\>|0\>+|0\>|\bm{u}\>|1\>)\).
    \label{Qcurvature_negaeigen_swap}
    \State Apply gate \({U}_H \otimes |0\>\< 0| + {V}_H\otimes |1\>\< 1|\) on the state to obtain \(\frac{1}{\sqrt{2}}(|\bm{P u}\>|0\>+ |\bm{Q u}\>|1\>)\).\label{Qcurvature_negaeigen_oracle}
    \State Apply the Hadmard gate on the third register to obtain the state \(\frac{|\bm{P u} \> + |\bm{Q u} \> }{2}|0\> + \frac{|\bm{Pu} \> - |\bm{Qu} \> }{2}|1\>\).\label{Qcurvature_negaeigen_add}
    \State Measure the third register and output the result.\label{Qcurvature_negaeigen_mea}
  \end{algorithmic}
\end{algorithm}

In Algorithm \ref{negative_eigenvalue}, \(\bm{P} \in \R^{d^2 \times d}\) is the matrix whose column vector \(\bm{p}_i = \bm{e}_i \otimes \frac{\bm{h}_i}{\|\bm{h}_i\|}\) for \(i \in [d]\) , and \(\bm{Q} \in \R^{d^2 \times d}\) is the matrix whose column vector \(\bm{q}_j = \frac{\bm{\tilde{h}}}{\|\bm{H}\|_F} \otimes \bm{e}_j\) for \(j \in [d]\). 
\(\bm{h}_i\) is the \(i\)-th column vector of matrix $\bm{H}$ and \(\tilde{\bm{h}}\) is a \(d\)-dimensional vector whose \(i\)-th component is \(\|\bm{h}_i\|\). 
It can be directly obtained that the matrix $\bm{P}$ and $\bm{Q}$ satisfy the decomposition \(\bm{H}/\|\bm{H}\|_F=\bm{P}^{T} \bm{Q}\) and have property \(\bm{P}^{T}\bm{P}=\bm{Q}^{T}\bm{Q}=I\). 
Mappings \(|\bm{x}\> |0\> \rightarrow |\bm{Px}\>\) and \(|0\> |\bm{x}\> \rightarrow |\bm{Qx}\>\) can be performed by the quantum oracle \(U_H\) and \(V_H\) respectively. 

\begin{theorem}
{\rm \textbf{Positive-negative eigenvalue discrimination.}}
\label{positive_negative}
For eigenvalue $\lambda$ of $\bm{H}$ with property $|\lambda| \geq a$, one could run Algorithm \ref{negative_eigenvalue} for $n=2[\frac{\|\bm{H}\|_F^2}{a^2}\log\frac{1}{\delta}-\frac{1}{2}]+3$ times, to make a statement that $\lambda$ is positive or negative, with probability \(1-\delta\).
\end{theorem}

The proof of Theorem \ref{positive_negative} is in the Appendix A.
Since the eigenvalue information of $\bm{H}$ is unknown to us, we need to build the Algorithm \ref{Label_proper} to label the \textbf{proper} eigenvalue, which would benefit the target state generation task in the following section. The \textbf{proper} eigenvalue means the eigenvalue is less than $-\alpha + \epsilon/2$. We view this kind of eigenvalue as our target eigenvalue.

 \begin{algorithm}[t]
  \caption{Proper Eigenvalue Labelling}
  \label{Label_proper}
  \begin{algorithmic}[1]
    \Require
      The Hessian matrix \(\bm{H} = \sum_{j=1}^{r} \lambda_j \bm{u}_j \bm{u}_j^{T}\) which is stored in the data structure in Section~\ref{QNCD_pre_technic_oracle}. The constant \(\alpha\) and the error parameter \(\epsilon\) in the NCF problem.
    \Ensure
    A \textbf{proper} label to the eigenvalue $\lambda_j$ such that \(\lambda_j \leq -\alpha+\epsilon/2\) with probability $1-\delta$, or a no-vector statement that there is no unit vector $\bm{u}$ which satisfies $\bm{u}^{T}\bm{H}\bm{u}<-\alpha$.
    \For{$k=1$ to $\frac{4\|\bm{H}\|_F^2}{\alpha^2}(2\frac{4\|\bm{H}\|_F^2}{\alpha^2}\log\frac{1}{\delta}+3)$}
    \State Create the state $\frac{1}{\|\bm{H}\|_F}\sum_{j=1}^{r} \lambda_j |\bm{u}_j\> |\bm{u}_j\>$.
    \label{Label_negative_init}
    \State Apply the quantum SVE model 
    to obtain the state $\frac{1}{\|\bm{H}\|_F}\sum_{j=1}^{r} \lambda_j |\bm{u}_j\> |\bm{u}_j\>||\tilde{\lambda}_j|\>$, where $|\tilde{\lambda}_j| \in |\lambda_j| \pm \epsilon/4$ with probability $1-1/\poly(d)$.
    \label{Label_negative_SVE}
    \State Measure the eigenvalue register and mark the result.
    \label{Label_negative_mea}
    \State Use the rest state in the first register as the input to apply the PNED algorithm.
    \label{Label_negative_judge}
    \EndFor
    \State Count the result in step \ref{Label_negative_mea} and step \ref{Label_negative_judge} to obtain the sequence $\{(|\tilde{\lambda}_j|, n_j, m_j)\}_{j=1}^{r}$\footnotetext{fvbfu}. $n_j$ is the number of resulting $|\tilde{\lambda}_j|$ in step \ref{Label_negative_mea}, and $m_j$ is the number of resulting $1$ in step \ref{Label_negative_judge} for different~$|\tilde{\lambda}_j|$.\label{Label_negative_count}
    \If{$\frac{m_j}{n_j} < \frac{1}{2}$ for all $j \in [r]$,} 
    \State make the no-vector statement;
    \Else
    \State choose the largest $|\tilde{\lambda}_j|$ which satisfies the condition $\frac{m_j}{n_j} > \frac{1}{2}$. 
    \If{$|\tilde{\lambda}_j|<\alpha-\epsilon/4$,} 
    \State make the no-vector statement;
    \Else
    \State label the eigenvalue $\lambda_j$ as the \textbf{proper} eigenvalue;
    \EndIf
    \EndIf
  \end{algorithmic}
\end{algorithm}

The mean idea of Algorithm \ref{Label_proper} is 
to use the input state: 
\begin{equation*}
	\frac{1}{\|\bm{H}\|_F}\sum_{j=1}^{r} \lambda_j |\bm{u}_j\> |\bm{u}_j\>,
\end{equation*}
to the quantum SVE model and obtain the state: 
\begin{equation*}
\frac{1}{\|\bm{H}\|_F}\sum_{j=1}^{r} \lambda_j |\bm{u}_j\> |\bm{u}_j\>||\tilde{\lambda}_j|\>.
\end{equation*}
The measurement on the eigenvalue register would let this entangled state collapse to different states $|\bm{u}_j\> |\bm{u}_j\>$ for $j \in [r]$. 
Since $|\bm{u}_j\> |\bm{u}_j\>$ is a pure state, we could obtain the state $|\bm{u}_j\>$ by neglecting the state in any other register. 
Using state $|\bm{u}_j\>$ to apply the PNED algorithm could provide a discrimination on the positive and negative of the corresponding eigenvalue $\lambda_j$. 
Thus we could label the \textbf{proper} eigenvalue or make the non-vector statement by the result of positive-negative discrimination and the measurement result on the eigenvalue register.

\begin{theorem}
\label{QNCF_claim_theorem}
Algorithm \ref{Label_proper} could label the \textbf{proper} eigenvalue of $\bm{H}$ with probability $1-1/\poly(d)$, or claim with high probability that there is no unit vector $\bm{u}$ which satisfies $\bm{u}^{T}\bm{H}\bm{u}<-\alpha$, with time complexity $\O(T_H \|\bm{H}\|_F^5 \polylog(d)\epsilon^{-1})$.
\end{theorem}

\begin{proof}
The input state $\frac{1}{\|\bm{H}\|_F}\sum_{j=1}^{r} \lambda_j |\bm{u}_j\> |\bm{u}_j\>$ could be generated with oracles $U_H$ and $V_H$:

\begin{equation}
\label{input_state}
|0\>|0\> \stackrel{V_H}{\longrightarrow} \frac{1}{\|\bm{H}\|_F} \sum_{i=1}^d\|\bm{h}_i\||i\>|0\> \stackrel{U_H}{\longrightarrow} \frac{1}{\|\bm{H}\|_F} \sum_{i=1}^d \sum_{j=1}^d h_{ij} |i\>|j\>.
\end{equation}

Since $\bm{H}$ has the eigen-decomposition $\bm{H}=\sum_{k=1}^{r} \lambda_k \bm{u}_k \bm{u}_k^{T}$, we could rewrite entry $h_{ij}$ as $h_{ij} = \sum_{k=1}^{r} \lambda_k u_k^{(i)} u_k^{(j)}$, where \(u_k^{(i)}\) is the \(i\)-th component of vector \(\bm{u}_k\). 
Thus the state $\frac{1}{\|\bm{H}\|_F} \sum_{i=1}^d \sum_{j=1}^d h_{ij} |i\>|j\>$ could be written as:
\begin{equation*}\frac{1}{\|\bm{H}\|_F} \sum_{i=1}^d \sum_{j=1}^d \sum_{k=1}^{r} \lambda_k u_k^{(i)} u_k^{(j)} |i\>|j\> = \frac{1}{\|\bm{H}\|_F}\sum_{k=1}^{r} \lambda_k |\bm{u}_k\> |\bm{u}_k\>.	
\end{equation*}
Then we apply the quantum SVE model on this state. In order to give $\epsilon/4$-estimation on the singular value, the time complexity to run the quantum SVE algorithm is $\O(T_H \|\bm{H}\|_F \polylog(d) \epsilon^{-1})$ by Theorem \ref{quantum_SVD}.

Suppose there are eigenvalues $\lambda_j$ which are less than $-\alpha+\epsilon/2$. We denote the least one as $\lambda_t$ and label it as the \textbf{proper} eigenvalue. By Theorem \ref{positive_negative}, we need to generate $n_t=2[\frac{\|\bm{H}\|_F^2}{\lambda_t^2}\log\frac{1}{\delta}-\frac{1}{2}]+3$ numbers of state $|\bm{u}_t\>$ in order to guarantee that $\lambda_t<0$ with probability $1-\delta$. 
Note that the probability of generating state \(|\bm{u}_t\>\) in each iteration of step \ref{Label_negative_init}-\ref{Label_negative_mea} in Algorithm \ref{Label_proper} is \(P_t=\frac{\lambda_t^2}{\|\bm{H}\|_F^2}\). 
So averagely we need to perform step \ref{Label_negative_init}-\ref{Label_negative_mea} in Algorithm \ref{Label_proper} for $n=\frac{\|\bm{H}\|_F^2}{\lambda_t^2}\{2[\frac{\|\bm{H}\|_F^2}{\lambda_t^2}\log\frac{1}{\delta}-\frac{1}{2}]+3\}$ times. 
The number $n$ can be roughly upper bounded by $\frac{4\|\bm{H}\|_F^2}{\alpha^2}(2\frac{4\|\bm{H}\|_F^2}{\alpha^2}\log\frac{1}{\delta}+3)$, since for negative curvature case $\epsilon < \alpha$, we have $|\lambda_t|=\alpha-\epsilon/2 > \alpha/2 $.

By considering the time complexity to run the quantum SVE algorithm ($\O(T_H \|\bm{H}\|_F \polylog(d) \epsilon^{-1})$) and setting the probability error bound $\delta = 1/\poly(d)$, we could derive that the time complexity of Algorithm \ref{Label_proper} is $\O(T_H \|\bm{H}\|_F^5 \polylog(d)\epsilon^{-1})$.
\end{proof}

\subsection{Target State Generating}
\label{3_input_state}

Suppose the result of Algorithm \ref{Label_proper} implies the existence of the target eigenvector \(\bm{u}_t\), which satisfies \(\bm{u}_t^{T} \bm{H} \bm{u}_t \le -\alpha+\epsilon/2\). In order to give a solution to the Negative Curvature Finding problem, we need to obtain the vector \(\bm{u}_t\) efficiently. Thus we develop Algorithm \ref{target_state} in Section \ref{3_input_state} which could generate the quantum state \(|\bm{u}_t\>\) in time $\tilde{\O}(\polylog(d))$. The classical read-out of \(|\bm{u}_t\>\), which means to estimate vector \(\bm{u}_t\) from quantum state \(|\bm{u}_t\>\), will be discussed in the following section.

 \begin{algorithm}[t]
  \caption{Target State Generating}
  \label{target_state}
  \begin{algorithmic}[1]
    \Require
The Hessian matrix \(\bm{H} = \sum_{j=1}^{r} \lambda_j u_j u_j^{T}\) which is stored in the data structure in Section~\ref{QNCD_pre_technic_oracle}. The number \(\alpha\) and the error bound \(\epsilon\) in NCF problem. The probability error bound \(\delta\).
    \Ensure
    The target state \(|\bm{u}_t\>\) with property \(\< \bm{u}_t|H|\bm{u}_t\> = \lambda_t \le -\alpha + \epsilon/2\).
    \For{$k=1$ to $\left[4\frac{\|\bm{H}\|_F^2}{\alpha^2}\log{\frac{1}{\delta}}\right]+1$}
    \State Create the state $\frac{1}{\|\bm{H}\|_F}\sum_{j=1}^{r} \lambda_j |\bm{u}_j\> |\bm{u}_j\>$.
    \label{target_state_input_init}
    \State Apply the quantum SVE model 
    to obtain the state $\frac{1}{\|\bm{H}\|_F}\sum_{j=1}^{r} \lambda_j |\bm{u}_j\> |\bm{u}_j\>||\tilde{\lambda}_j|\>$, where $|\tilde{\lambda}_j| \in |\lambda_j| \pm \epsilon/4$ with probability $1-1/\poly(d)$.
    \label{target_state_SVE}
    \State Measure the eigenvalue register and mark the result.
    \label{target_state_mea}
    \If {the eigenvalue measured in in step \ref{target_state_mea} is labelled to be \textbf{proper} in Algorithm \ref{Label_proper},}
    \State output the state in the first register as the target state.    \label{target_state_output}
    \EndIf \label{target_state_endif}
    \EndFor
  \end{algorithmic}
\end{algorithm}

The main idea of Algorithm \ref{target_state} is very similar to Algorithm \ref{Label_proper}. We still use state $\frac{1}{\|\bm{H}\|_F}\sum_{j=1}^{r} \lambda_j |\bm{u}_j\> |\bm{u}_j\>$ as the input of quantum SVE algorithm to obtain state:
\begin{equation*}
\frac{1}{\|\bm{H}\|_F}\sum_{j=1}^{r} \lambda_j |\bm{u}_j\> |\bm{u}_j\>||\tilde{\lambda}_j|\>.
\end{equation*}
Suppose $\lambda_t$ denotes the eigenvalue of \(|\bm{u}_t\>\) that \(\lambda_t \le -\alpha+\epsilon/2\). 
The probability of generating state \(|\bm{u}_t\>\) in each iteration of step \ref{target_state_input_init}-\ref{target_state_endif} in Algorithm \ref{target_state} is \(P_t =  \frac{\lambda_t^2}{\|\bm{H}\|_F^2} \ge \frac{\alpha^2}{4\|\bm{H}\|_F^2}\). 
Thus the probability of generating at least one state \(|\bm{u}_t\>\) in $N=\left[4\frac{\|\bm{H}\|_F^2}{\alpha^2}\log{\frac{1}{\delta}}\right]+1$ times of step \ref{target_state_input_init}-\ref{target_state_endif} is \(1-(1-P_t)^N\). There is: 
\begin{equation*}
1-(1-P_t)^N \ge 1-e^{-N P_t} \ge 1-e^{-\log(1/\delta)} = 1-\delta.
\end{equation*}
So Algorithm \ref{target_state} could generate at least one state \(|\bm{u}_t\>\) in $N$ iterations with probability at least \(1-\delta\).
By considering the time complexity to run the quantum SVE algorithm ($\O(T_H\|\bm{H}\|_F \polylog(d) \epsilon^{-1})$) and setting the probability error bound $\delta = 1/\poly(d)$, we could derive the time complexity of Algorithm \ref{target_state} in Theorem~\ref{state_finding}:

\begin{theorem}
\label{state_finding}
Suppose that the target state which satisfies the condition \(\< \bm{u}_t|\bm{H}|\bm{u}_t\> \le -\alpha + \epsilon/2\) exists. There is a quantum algorithm which could perform this target state \(|\bm{u}_t\>\) in time \(\O(T_H \|\bm{H}\|_F^3 \polylog(d) \epsilon^{-1})\) with probability at least \(1-1/\poly(d)\).
\end{theorem}

\section{State Read-out}
\label{QNCD_state_read-out}

In this section, we propose an efficient algorithm to readout the classical vector $\bm{u}_t$ from the quantum state $|\bm{u}_t\>$. Generally, the classical read-out of a quantum state takes at least \(\O(d/\epsilon_1)\) times of measurement on $d$-dimensional quantum state for an $\epsilon_1$-error estimation. Thus the classical read-out of the required state could offset the exponential speed-up \cite{Aaronson2015QuantumML} provided in many quantum machine learning algorithms. In order to avoid this problem, we propose Algorithm \ref{coordinate} to rewrite the target state $|\bm{u}_t\>$ as the linear combination of $r$ states, which is selected from column vectors of Hessian $\bm{H}$.

Recall that our Hessian matrix \(\bm{H} \in \R^{d \times d}\) has the eigendecomposition \(\bm{H} = \sum_{j=1}^{r} \lambda_j \bm{u}_j \bm{u}_j^{T}\). The eigenfunction \(\bm{H} \bm{u}_j = \lambda_j \bm{u}_j \) can be written as \(\sum_{i=1}^{d} \bm{h}_i u_j^{(i)}=\lambda_j \bm{u}_j\), which means any eigenvector of $\bm{H}$ that corresponds to a non-zero eigenvalue could be represented as the linear combination of vectors in  \(\{\bm{h}_j\}_{j=1}^{d}\). 
Since $\bm{H}$ has the rank of \(r\), there exists a subset of complete basis \(\{\bm{h}_{g(i)}\}_{i=1}^{r}\)\footnote{$g(i)$ is the index of the $i$-th column vector in the complete basis.} of the column space, which is sampled from the set \(\{\bm{h}_j\}_{j=1}^{d}\). 
Thus, any eigenvector $\bm{u}_j$ could also be represented as the linear combination of vectors in \(\{\bm{h}_{g(i)}\}_{i=1}^{r}\). We denote $\bm{s}_i$ as $\bm{h}_{g(i)}$ for simplicity.

Back to the state read-out problem, suppose the target state \(|\bm{u}_t\>\) that we generated in previous section can be written as \(|\bm{u}_t\> = \sum_{i=1}^{r} x_i |\bm{h}_{g(i)}\>\), where \(\{x_i\}_{i=1}^{r}\) are coordinates of state \(|\bm{u}_t\>\) under the basis \(\{|\bm{h}_{g(i)}\>\}_{i=1}^{r}\). Thus, instead of simply reading out components of vector \(\bm{u}_t\), we could get the classical description of \(|\bm{u}_t\>\) by calculating each \(x_i\). Note that the complete basis \(\{|\bm{h}_{g(i)}\>\}_{i=1}^{r}\) is not unique and we only need to identify one of them.

One advantage of generating the form $|\bm{u}_t\>=\sum_{i=1}^{r} x_i |\bm{s}_i\>$ is that the updating operation for the $l$-th iteration $\bm{z}^{(l+1)}= \bm{z}^{(l)} +\eta \bm{u}_t$ in general non-convex algorithms could be implemented in quantum form $|\bm{z}^{(l+1)}\> \propto |\bm{z}^{(l)}\> +\eta \sum_{i=1}^{r} x_i |\bm{s}_i\>$ by Linear-Sum-of-States method \cite{shao2018linear} efficiently, which may inspire efficient quantum counterparts for these non-convex algorithms. Note that we could not perform the operation $|\bm{z}^{(l+1)}\> \propto |\bm{z}^{(l)}\> +\eta |\bm{u}_t\>$ directly since the target state $|\bm{u}_t\>$ is generated by post-selection instead of standard unitary operations.


\subsection{Complete Basis Selection}
\label{4_basis}

In this section, we develop a quantum algorithm to select a subset \(S_I=\{g(1), g(2),\cdots, g(r)\}\) from \([d]\), which corresponds to the complete basis \(\{|\bm{h}_{g(i) }\>\}_{i=1}^{r}\). 
The quantum complete basis selection algorithm can be viewed as the quantum version of Gram-Schmidt orthogonalization: firstly we choose $|\bm{t}_1\>=|\bm{h}_{g(1)}\>$ from the state set $\{|\bm{h}_j\>\}_{j=1}^d$; then given state set $\{|\bm{t}_m\>\}_{m=1}^{l}$, we choose $|\bm{t}_{l+1}\> \propto
|\bm{h}_{g(l+1)}\>-\sum_{m=1}^{l} |\bm{t}_m\>\<\bm{t}_m|\bm{h}_{g(l+1)}\>$ from the state set $\{|\bm{h}_j\>-\sum_{m=1}^{l} |\bm{t}_m\>\<\bm{t}_m|\bm{h}_j\>\}_{j=1}^{d}$. 
Since the chosen $|\bm{t}_{l+1}\>$ is orthogonal to states in basis $\{|\bm{t}_m\>\}_{m=1}^{l}$ for each iteration $l \in [r-1]$, state set $\{|\bm{t}_m\>\}_{m=1}^{r}$ forms an orthonormal basis. 

Note that the state $|\bm{t}_m\>$ is generated along with an index $g(m)$ for $m \in [r]$, so we would obtain a complete basis index set $\{g(m)\}_{m=1}^{r}$ after the implementation of this quantum algorithm. More detail about the quantum complete basis selection algorithm is provided in Algorithm \ref{Basis}. The detail of the time complexity of Algorithm \ref{Basis} is analyzed in Theorem \ref{Basis_theorem} in Section \ref{4.2_error}.

\begin{algorithm}[t]
\caption{Complete Basis Selection}
\label{Basis}
\begin{algorithmic}[1]
\Require
Quantum access to oracle $U_H$ and $V_H$. 
\Ensure
The index set of the complete basis: $S_I = \{g(i)\}_{i=1}^{r}$.
\State Initialize the index set $S_I= \varnothing $.\label{Basis_one}
\For{$l=0$ to $r-1$}
\State Create state $|\phi_1^{(l)}\>=\frac{1}{\|\bm{H}\|_F} \sum_{j=1}^d \|\bm{h}_j\| |j\> \left\{ \left[|\bm{h}_j\>-\sum_{m=1}^{l}|\bm{t}_m\>\<\bm{t}_m|\bm{h}_j\>\right] |{0}\> - \sum_{m=1}^{l}|\bm{t}_m\>\<\bm{t}_m|\bm{h}_j\>|1\> \right\}$.\label{Basis_state}
\State Measure the third register of state $|\phi_1^{(l)}\>$ multiple times to get the state $|\phi_2^{(l)}\>$ which is proportional to $\frac{1}{\|\bm{H}\|_F}\sum_{j=1}^{d} |j\>
\|\bm{h}_j\|\left[|\bm{h}_j\>-\sum_{m=1}^{l} |\bm{t}_m\>\<\bm{t}_m|\bm{h}_j\>\right]$.\label{Basis_measure}
\State Measure the first register and record the result as $g(l+1)$. \label{Basis_record}
\State Denote $|t_{l+1}\>$ as the state proportional to $|\bm{h}_{g(l+1)}\>-\sum_{m=1}^{l} |\bm{t}_m\>\<\bm{t}_m|\bm{h}_{g(l+1)}\>$. \label{Basis_basis}
\State Update the index set $S_I = S_I \cup \{g(l+1)\}$.\label{Basis_index}
\EndFor
\end{algorithmic}
\end{algorithm}

State $\frac{1}{\|\bm{H}\|_F} \sum_{j=1}^d \|\bm{h}_j\| |j\> \left\{ \left[|\bm{h}_j\>-\sum_{m=1}^{l}|\bm{t}_m\>\<\bm{t}_m|\bm{h}_j\>\right] |{0}\> - \sum_{m=1}^{l}|\bm{t}_m\>\<\bm{t}_m|\bm{h}_j\>|1\> \right\}$ in Step~\ref{Basis_state} of Algorithm \ref{Basis} can be generated by the following procedure:
\small
\begin{align*}
	|0\>|0\>
	\xrightarrow{U_H V_H} 
	&\frac{1}{\|\bm{H}\|_F}\sum_{j=1}^{d} \|\bm{h}_j\||j\>|\bm{h}_j\>\\
	\xrightarrow{\text{add an auxiliary register and apply Hadamard gate}}
	&\frac{1}{\|\bm{H}\|_F}\sum_{j=1}^{d} \|\bm{h}_j\||j\>|\bm{h}_j\>\frac{|0\>+|1\>}{\sqrt{2}}\\
	\xrightarrow{ \prod_{m=1}^{l}\left[(I-2|\bm{t}_m\>\<\bm{t}_m|)\otimes |0\>\<0|+I\otimes |1\>\<1|\right]}
	&\frac{1}{\|\bm{H}\|_F} \sum_{j=1}^d \|\bm{h}_j\| |j\> \left\{ \left[ |\bm{h}_j\>-2\sum_{m=1}^{l}|\bm{t}_m\>\<\bm{t}_m|\bm{h}_j\> \right] \frac{|0\>}{\sqrt{2}}+|\bm{h}_j\> \frac{|1\>}{\sqrt{2}}\right\}\\
	\xrightarrow{\text{apply Hadamard gate on the last register}}& \frac{1}{\|\bm{H}\|_F} \sum_{j=1}^d \|\bm{h}_j\| |j\> \left\{ \left[|\bm{h}_j\>-\sum_{m=1}^{l}|\bm{t}_m\>\<\bm{t}_m|\bm{h}_j\>\right] |{0}\> - \sum_{m=1}^{l}|\bm{t}_m\>\<\bm{t}_m|\bm{h}_j\>|1\> \right\}. 
\end{align*}
\normalsize

The crucial part in Algorithm~\ref{Basis} is to implement the reflection $R_m = I-2|\bm{t}_m\>\<\bm{t}_m|, \forall m \in [r-1]$. For the $m+1$ case, there is:
\begin{equation}\label{4.1_t_m+1}
|\bm{t}_{m+1}\>=\frac{1}{Z_{m+1}}(|\bm{s}_{m+1}\>-\sum_{i=1}^{m} |\bm{t}_i\>\<\bm{t}_i|\bm{s}_{m+1}\>),	
\end{equation}
Note that $\{|\bm{t}_i\>\}$ forms the orthogonal basis: $\<\bm{t}_i|\bm{t}_j\>=0, \forall i \neq j$, and $Z_{m+1}=\<\bm{t}_{m+1}|\bm{s}_{m+1}\>=\||\bm{s}_{m+1}\>-\sum_{i=1}^{m} |\bm{t}_i\>\<\bm{t}_i|\bm{s}_{m+1}\>\|=\sqrt{1-\sum_{i=1}^{m} \<\bm{t}_i|\bm{s}_{m+1}\>^2}$.

Define coordinate $\{x_{ij}\}$ such that each state $|\bm{t}_i\>$ could be written as $\sum_{j=1}^{i} x_{ij} |\bm{s}_j\>$. Since $|\bm{t}_{m+1}\>$ is orthogonal to $|\bm{s}_1\>,|\bm{s}_2\>,\cdots|\bm{s}_m\>$ and required to be normal, there is:
\begin{equation}\label{4.1_x_1}
\left\{
\begin{aligned}
&\sum_{i=1}^{m+1} x_{m+1,i} \<\bm{s}_j|\bm{s}_i\> = 0,\ \forall j \in [m],\\	
& \sum_{j=1}^{m+1} \sum_{i=1}^{m+1} x_{m+1,j} x_{m+1,i} \<\bm{s}_j|\bm{s}_i\> = 1 .
\end{aligned}
\right.
\end{equation}
Note that $x_{m+1,m+1}=1/Z_{m+1}$ by equation (\ref{4.1_t_m+1}). Define $m$-dimensional vector $\boldsymbol{x}, \boldsymbol{b}$ and the $m \times m$ matrix $\boldsymbol{C}_m$, such that:
\begin{equation*}
\boldsymbol{x}=\sum_{i=1}^m x_{m+1,i} \boldsymbol{e}_i\ ,
\boldsymbol{b}=\sum_{i=1}^m \<\bm{s}_j|\bm{s}_{m+1}\>\boldsymbol{e}_j\ ,
\boldsymbol{C}_m=\{c_{ij}\}_{ij}^{m,m}=\{\<\bm{s}_i|\bm{s}_j\>\}_{ij}^{m,m}\ .
\end{equation*}

Thus, we could derive the following linear equations about $\boldsymbol{x}$:
\begin{equation}\label{4.1_x_2}
\left\{
\begin{aligned}
&\boldsymbol{C}_m \boldsymbol{x}=-\frac{1}{Z_{m+1}}\boldsymbol{b},\\	
& \boldsymbol{x}^{T} \boldsymbol{b}=\frac{Z_{m+1}^2-1}{Z_{m+1}}.
\end{aligned}
\right.
\end{equation}

We could obtain the coordinate $\{x_{m+1,i}\}_{i=1}^{m+1}$ by solving equation~\ref{4.1_x_2}. There is:
\begin{equation}\label{4.1_x_3}
\left\{
\begin{aligned}
&x_{m+1,m+1}=\frac{1}{Z_{m+1}}=\sqrt{\frac{|\boldsymbol{C}_{m}|}{|\boldsymbol{C}_{m+1}|}},\\	
& x_{m+1,i}=-\frac{1}{Z_{m+1}}\frac{|\boldsymbol{C}_m^{(i)}|}{|\boldsymbol{C}_m|},\ \forall i \in [m],
\end{aligned}
\right.
\end{equation}
where matrix $\boldsymbol{C}_m^{(i)}$ denotes the matrix generated from $\boldsymbol{C}_m$ by replacing the $i$-th column with $\boldsymbol{b}$.

Suppose now we have obtained the linear combination form $|\bm{t}_{m+1}\>=\sum_{i=1}^{m+1} x_{m+1,i}|\bm{s}_i\>$. Thus, in order to prepare state $|\bm{t}_{m+1}\>$, we perform the states-linear-sum operation.

Consider the binary tree structure in Figure~\ref{4.1_tree_t_m+1}.
\begin{figure}
\centering
\begin{tikzpicture}[grow'=up, level distance=1.3cm,
   level 1/.style={sibling distance=8cm, level distance=1.6cm},
   level 2/.style={sibling distance=4cm, level distance=1.2cm},level 3/.style={sibling distance=2cm, level distance=0.8cm}]
\node {$|\bm{t}_{m+1}\>$}
child{node {$\cdots$}
   child {node {$x_{m+1,1}|\bm{s}_1\>+x_{m+1,2}|\bm{s}_2\>$}
     child {node {$|\bm{s}_1\>$}}
     child {node {$|\bm{s}_2\>$}}
   }
   child {node {$\cdots$}
     child {node {$\cdots$}}
     child {node {$\cdots$}}
   }
}
child {node {$\cdots$}
   child {node {$x_{m+1,m}|\bm{s}_{m}\>+x_{m+1,m+1}|\bm{s}_{m+1}\>$}
     child {node {$|\bm{s}_m\>$}}
     child {node {$|\bm{s}_{m+1}\>$}}
   }
   child{
   }
};
\end{tikzpicture}
\caption{The structure of generating $|\bm{t}_{m+1}\>$ by linear-combination-of-states method}
\label{4.1_tree_t_m+1}
\end{figure}
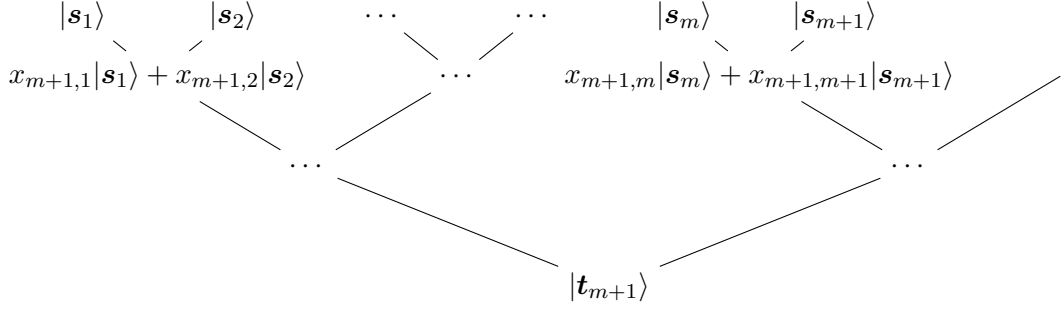
The required state $|\bm{t}_{m+1}\>$ could be generated by performing the states-linear-sum operation for each child-parent-child structure in Figure~\ref{4.1_tree_t_m+1} from the leave layer to the root layer. Note that here we actually perform the state $|\bm{t}_{m+1}'\>=|g(1)\>|g(2)\>\cdots|g(m+1)\>|\bm{t}_{m+1}\>$ from states $\{|g(1)\>|g(2)\>\cdots|g(m+1)\>|s_i\>\}_{i=1}^{m+1}$, and each state $(\prod_{j=1}^{m+1} |g(j)\>)|s_i\>$ could be performed by oracle $U_H$ on state $(\prod_{j=1}^{m+1} |g(j)\>)|0\>$. The operation:
\begin{equation*}
\prod_{m=1}^{l}\left[(I-2|\bm{t}_m\>\<\bm{t}_m|)\otimes |0\>\<0|+I\otimes |1\>\<1|\right]
\end{equation*}
is performed by apply $\prod_{m=1}^{l}\left[(I-2|\bm{t}_m'\>\<\bm{t}_m'|)\otimes I \otimes |0\>\<0|+I\otimes |1\>\<1|\right]$ on state:
\begin{equation*}
\frac{1}{\|\bm{H}_F\|}\sum_{i=1}^{d}\|\bm{h}_i\|(\prod_{j=1}^{l} |g(j)\>)|\bm{h}_i\>\frac{|0\>+|1\>}{\sqrt{2}}.
\end{equation*}

We neglect the term $\prod_{j=1}^{l} |g(j)\>$ for each $l$ iteration in Algorithm \ref{Basis} for simplicity, which do not influence the result.

\subsection{Error analysis and time complexity}
\label{4.2_error}
The error of implementing $|\bm{t}_{m+1}\>$ comes from the imperfect implementing of rotation operation for each child-parent-child structure in Figure~\ref{4.1_tree_t_m+1}, and the error of calculating $\boldsymbol{x}$. The former has been analyzed in \cite{shao2018linear}:
\begin{theorem}{\rm \cite{shao2018linear}}\label{4.2_linear_state}
Assume state $|\phi_i\>$ could be prepared by given unitary operation in time $\O(T_{in})$, for $i \in [n]$. Then there is a unitary which could prepare the state $|\phi\>=\sum_{i=1}^n \alpha_i |\phi_i\>$ in time $\O(T_{in} n^{\log(n/\epsilon)})$ with error $\epsilon$.
\end{theorem}

The error of calculating $\boldsymbol{x}$ is more complex. Define vector $\boldsymbol{y}=-Z_{m+1}\boldsymbol{x}$. 
Since all parameters $c_{ij}=\<\bm{s}_i|\bm{s}_j\>, i,j\in [m]$ and $b_j=\<\bm{s}_{m+1}|\bm{s}_j\>, j \in [m]$ are estimated by Hadamard test (see Appendix B for more information), the vector $\boldsymbol{y}=\boldsymbol{C}_m^{-1}\boldsymbol{b}$ would have an error. 
Define matrix $\boldsymbol{\tilde{C}}_m=\boldsymbol{C}_m+\Delta\boldsymbol{C}_m$ and $\boldsymbol{\tilde{b}}=\boldsymbol{b}+\Delta\boldsymbol{b}$ which are estimations on $\boldsymbol{C}_m$ and $\boldsymbol{b}$. 
Suppose $|\tilde{c}_{ij} - c_{ij}| \leq \epsilon_1$ and $|\tilde{b}_i -b_i |\leq \epsilon_1$ are error bounds for $c_{ij}$ and $b_i$, respectively, $\forall i,j \in [m]$.
Denote $\boldsymbol{\tilde{y}}=\boldsymbol{\tilde{C}}_m^{-1}\boldsymbol{b}$ as the solution to the approximate linear equation and $\Delta \boldsymbol{y}=\boldsymbol{\tilde{y}}-\boldsymbol{y}$ as the error to $\boldsymbol{y}$. We have:
\begin{equation*}
\left\{
\begin{aligned}
&\boldsymbol{C}_m \boldsymbol{y}=\boldsymbol{b},\\
&(\boldsymbol{C}_m +\Delta \boldsymbol{C}_m)(\boldsymbol{y}+ \Delta \boldsymbol{y})=(\boldsymbol{b}+\Delta \boldsymbol{b}).
\end{aligned}
\right.
\end{equation*}
So there is:
\begin{align*}
\|\D \y\| &= \|(\C_m+\D\C_m)^{-1}(\D\b-\D\C_m\cdot\C_m^{-1}\b)\| \\
&\leq \|\C_m^{-1}\| \cdot \| (\boldsymbol{I}+\C_m^{-1}\D\C_m)^{-1}\|\cdot (\|\D\b\|+\|\D\C_m \cdot \C_m^{-1}\b\|) \\
& \leq \|\C_m^{-1}\|\cdot\frac{1}{1-\|\C_m^{-1}\D\C_m\|}\cdot(\|\D\b\|+\|\D\C_m\|\|\C_m^{-1}\|\|\b\|)\\
& \leq \frac{\|\C_m^{-1}\|}{1-\|\C_m^{-1}\| m \epsilon_1}\cdot(\sqrt{m}\epsilon_1+m^{3/2} \epsilon_1\|\C_m^{-1}\|) \leq \frac{2m^{3/2}\|\C_m^{-1}\|^2 \epsilon_1}{1-\|\C_m^{-1}\|m\epsilon_1}.
\end{align*}

The norm $\|\C_m^{-1}\|$ here denotes the largest norm of eigenvalues of matrix $\C_m^{-1}$.
Thus, for $\epsilon_2=3m^{3/2}\|\C_m^{-1}\|^2 \epsilon_1$, we have bound $\|\D\y\| \leq \epsilon_2$.

For $|\bm{t}_{m+1}\>=\sum_{i=1}^{m+1}x_{i}|\bm{s}_i\>$ and $|\tilde{\bm{t}}_{m+1}\>=\sum_{i=1}^{m+1}\tilde{x}_{i}|\bm{s}_i\>$, there is:
\begin{align*}
{\<\bm{t}_{m+1}|\tilde{\bm{t}}_{m+1}\>} &=
\sum_{i=1}^{m+1}\sum_{j=1}^{m+1}x_i c_{ij} \tilde{x}_j = \frac{1}{Z_{m+1}\tilde{Z}_{m+1}}[\sum_{i=1,j=1}^{m,m} y_i c_{ij} \tilde{y}_j -\sum_{i=1}^{m} y_i b_i -\sum_{j=1}^{m} b_j \tilde{y}_j +1] \\
&= \frac{1}{Z_{m+1}\tilde{Z}_{m+1}}[1-\sum_{i=1}^{m} y_i b_i]=\frac{Z_{m+1}}{\tilde{Z}_{m+1}}.
\end{align*}
Note that the form $|\tilde{\bm{t}}_{m+1}\>=\sum_{i=1}^{m+1}\tilde{x}_{i}|\bm{s}_i\>$ is not a normalized state: 
\begin{align*}
\||\tilde{\bm{t}}_{m+1}\>\|^2 &= \sum_{i=1}^{m+1}\sum_{j=1}^{m+1}\tilde{x}_i c_{ij} \tilde{x}_j = \frac{1}{\tilde{Z}_{m+1}^2} [\sum_{i=1,j=1}^{m,m} \tilde{y}_i c_{ij} \tilde{y}_j +1-2\sum_{i=1}^{m} \tilde{y}_i b_i] \\
&=\frac{1}{\tilde{Z}_{m+1}^2}[\sum_{i=1,j=1}^{m,m} \Delta y_i c_{ij} \Delta y_j +2\sum_{i=1,j=1}^{m,m}\Delta y_i c_{ij} y_j + \sum_{i=1,j=1}^{m,m} y_i c_{ij} y_j +1 -2\sum_{i=1}^{m} \Delta y_i b_i -2\sum_{i=1}^m y_i b_i]\\
&=\frac{1}{\tilde{Z}_{m+1}^2} [Z_{m+1}^2 + \sum_{i=1,j=1}^{m,m} \Delta y_i c_{ij} \Delta y_j].
\end{align*}
Thus, the overlap between state $|\bm{t}_{m+1}\>$ and $|\tilde{\bm{t}}_{m+1}\>$ is:
\begin{equation}
\frac{\<\bm{t}_{m+1}|\tilde{\bm{t}}_{m+1}\>}{\||\tilde{\bm{t}}_{m+1}\>\|} = \frac{Z_{m+1}}{\sqrt{Z_{m+1}^2 + \D\y^{T} \C_m\D\y}}.
\end{equation}

There is:
\begin{align}
\||\bm{t}_{m+1}\>-|\tilde{\bm{t}}_{m+1}\>/\||\tilde{\bm{t}}_{m+1}\>\|\| &= \sqrt{2-2 \frac{Z_{m+1}}{\sqrt{Z_{m+1}^2 + \D\y^{T} \C_m\D\y}}}\\
& \leq \sqrt{2-2\frac{Z_{m+1}}{\sqrt{Z_{m+1}^2 +\|\D\y\|^2\|\C_m\|}}}\\
&\leq \sqrt{2-2\frac{Z_{m+1}}{Z_{m+1}+\frac{\|\C_m\|\|\D\y\|^2}{2Z_{m+1}}}}\\
&\leq \frac{\|\C_m\|^{1/2}\|\D\y\|}{Z_{m+1}}\\
&\leq \frac{m^{1/2}\epsilon_2}{Z_{m+1}}=\frac{3m^2 \|\C_m^{-1}\|^2 \epsilon_1}{Z_{m+1}}.
\end{align}

Let $\epsilon_3/2=\frac{3m^2 \|\C_m^{-1}\|^2 \epsilon_1}{Z_{m+1}}$. It is clear that to obtain the state form $|\tilde{\bm{t}}_{m+1}\>=\sum_{i=1}^{m+1}\tilde{x}_{m+1,i}|\bm{s}_i\>$ takes time $\O(m^3 +m^2 T_H\epsilon_1^{-2})$, where $\||\tilde{\bm{t}}_{m+1}\>/\||\tilde{\bm{t}}_{m+1}\>\|-|\bm{t}_{m+1}\>\| \leq \epsilon_3/2$. 
By Theorem \ref{4.2_linear_state}, the implementation of state $|\tilde{\bm{t}}_{m+1}\>$ takes time $\O(T_H(m+1)^{\log(2(m+1)/\epsilon_3)})$ with error bounds $\epsilon_3/2$. 
Thus, we can implemente state $|\bm{t}_{m+1}\>$ by unitary in time $\O(m^3 +m^2 T_H\epsilon_1^{-2}+T_H(m+1)^{\log(2(m+1)/\epsilon_3)})$ with error bounds in $\epsilon_3$.

Now we consider the influence of imperfect implementation on state $|\tilde{\bm{t}}_m\>$ to Algorithm~\ref{Basis}. 
Denote $R_l=\prod_{i=1}^{l}(I-2|\bm{t}_i\>\<\bm{t}_i|)$ and $\tilde{R}_l=\prod_{i=1}^{l} (I-2|\tilde{\bm{t}}_{i}\>\<\tilde{\bm{t}}_i|)$. There is $\|R_l-\tilde{R}_l\|\leq 2l\epsilon_3$ by \cite{nielsen2002quantum}.

Note that the state 
\begin{equation*}	
|\phi_1^{(l)}\>=\frac{1}{\|\bm{H}\|_F} \sum_{j=1}^d \|\bm{h}_j\| |j\> \left\{ \left[|\bm{h}_j\>-\sum_{m=1}^{l}|\bm{t}_m\>\<\bm{t}_m|\bm{h}_j\>\right] |{0}\> - \sum_{m=1}^{l}|\bm{t}_m\>\<\bm{t}_m|\bm{h}_j\>|1\> \right\}
\end{equation*}
in step \ref{Basis_state} of Algorithm \ref{Basis} can be written as:
\begin{equation*}
	\frac{1}{\|\bm{H}\|_F}\sum_{j=1}^{d}\|\bm{h}_j\||j\>[\frac{R_l+I}{2}|\bm{h}_j\>|0\>+\frac{R_l-I}{2}|\bm{h}_j\>|1\>].
\end{equation*}

The probability of generating $0$ after the measurement on the last register is:
\begin{equation*}
P_l = \frac{1}{\|\bm{H}\|_F^2}\sum_{j=1}^d \|\bm{h}_j\|^2\|\frac{R_l +I}{2}|\bm{h}_j\>\|^2.
\end{equation*}

Similarly we define 
$|\tilde{\phi}_1^{(l)}\>=\frac{1}{\|\bm{H}\|_F}\sum_{j=1}^{d}\|\bm{h}_j\||j\>[\frac{\tilde{R}_l+I}{2}|\bm{h}_j\>|0\>+\frac{\tilde{R}_l-I}{2}|\bm{h}_j\>|1\>]$,

 and $\tilde{P}_l=\frac{1}{\|\bm{H}\|_F}\sum_{j=1}^d \|\bm{h}_j\|^2\|\frac{\tilde{R}_l +I}{2}|\bm{h}_j\>\|^2$ for the approximate case.
 
Since the objective of step \ref{Basis_record} is to obtain the index $g(l+1)$ such that the column state $|\bm{h}_{g(l+1)}\>=|\bm{s}_{l+1}\>$ is independent from $\{|\bm{s}_i\>\}_{i=1}^{l}$, we define $P_l^{false}$ as the probability of selecting out the state $|\bm{s}_{l+1}\> \in \{|\bm{s}_i\>\}_{i=1}^{l}$ in the approximate case.
Note that for state $|\bm{h}_j\> \in \{|\bm{s}_i\>\}_{i=1}^{l}$, $(R_l +I)|\bm{h}_j\>=0$, so there is:
\begin{align*}
P_l^{false} &= \frac{1}{\tilde{P}_l}\frac{1}{\|\bm{H}\|_F^2} \sum_{j:|\bm{h}_j\> \in \{|\bm{s}_i\>\}_{i=1}^{l}} \|\bm{h}_j\|^2\|\frac{\tilde{R}_l+I}{2}|\bm{h}_j\>\|^2 \\
& =\frac{\sum_{j:|\bm{h}_j\> \in \{|\bm{s}_i\>\}_{i=1}^{l}} \|\bm{h}_j\|^2\|\frac{\tilde{R}_l+I}{2}|\bm{h}_j\>\|^2}{\sum_{j=1}^d \|\bm{h}_j\|^2\|\frac{\tilde{R}_l+I}{2}|\bm{h}_j\>\|^2}\\
&= \frac{\sum_{j:|\bm{h}_j\> \in \{|\bm{s}_i\>\}_{i=1}^{l}} \|\bm{h}_j\|^2\|\frac{\tilde{R}_l-R_l}{2}|\bm{h}_j\>\|^2}{\sum_{j=1}^d \|\bm{h}_j\|^2 (1/2+\<\bm{h}_j|\tilde{R}_l-R_l |\bm{h}_j\>/2+\<\bm{h}_j|R_l|\bm{h}_j\>/2)}\\
&\leq \frac{l^2 \epsilon_3^2}{P_l-l\epsilon_3}
\leq \frac{l^2 \epsilon_3^2}{\frac{(r-l)\epsilon^2}{4\|\bm{H}\|_F^2}-l\epsilon_3}.
\end{align*}

Let $\epsilon_3=\frac{\epsilon^2}{8(r-1)\|\bm{H}\|_F^2}$, there is:
\begin{align}\label{successAlgorithm}
\sum_{l=0}^{r-1}P_l^{false} \leq \sum_{l=0}^{r-1}	\frac{l^2 \epsilon_3^2}{2(r-1)(r-l)\epsilon_3-l\epsilon_3} \leq \sum_{l=0}^{r-1} l \epsilon_3 =\frac{r(r-1)}{2}\frac{\epsilon^2}{8(r-1)\|\bm{H}\|_F^2} =\frac{r\epsilon^2}{16\|\bm{H}\|_F^2} \leq \frac{1}{4}.
\end{align}

Thus, by choosing $\epsilon_3=\frac{\epsilon^2}{8(r-1)\|\bm{H}\|_F^2}$, Algorithm \ref{Basis} could select out a complete basis $\{|\bm{s}_i\>\}_{i=1}^{r}$ with probability at least $\frac{3}{4}$.

Note that in Algorithm \ref{Basis} we need to perform operations $R_i=I-2|\bm{t}_i\>\<\bm{t}_i|$ for $i=1,2,\cdots,r-1$, which needs the information of parameters $\{c_{ij}\}_{i=1,j=1}^{r-1,r-1}$. In order to guarantee the success probability of Algorithm \ref{Basis} (equation (\ref{successAlgorithm})), the estimation on each $c_{ij}=\<\bm{s}_i|\bm{s}_j\>$ should have error bound $\epsilon_1= \min_{m\in [r-2]} \frac{Z_{m+1}}{6m^2 \|\C_m^{-1}\|^2}\epsilon_3 \geq \frac{\min_{m \in [r-1]}Z_m}{48r^3\|\C_r^{-1}\|^2\|\bm{H}\|_F^2}\epsilon^2$. 
Thus, the estimation on each $c_{ij}$ takes time $\O(T_H r^6 \|\bm{H}\|_F^4\epsilon^{-4})$ and the estimation on the parameter group $\{c_{ij}\}_{i=1,j=1}^{r-1,r-1}$ takes time $\O(T_H r^8 \|\bm{H}\|_F^4\epsilon^{-4})$. 
Since the state $|\bm{t}_m\>$ is implemented by Linear-Sum-of-States $|t_{m}\>=\sum_{i=1}^{m}x_{m,i}|\bm{s}_i\>$, additional time is required to solve $m-1$-dimensional equations for $m=2,3,\cdots r-1$, which results the time complexity $\O(r^4)$ in total. With given parameters $\{x_{m,i}\}_{i=1}^m$, the implementation of operation $R_m=I-2|\bm{t}_m\>\<\bm{t}_m|$ takes time $\O(T_H(m+1)^{\log(2(m+1)/\epsilon_3)})$ $\leq$ $\O(T_H(r)^{2\log(4r\|\bm{H}\|_F/\epsilon)})$. 

Denote $P_l$ as the probability of resulting ${0}$ after the measurement in Step \ref{Basis_measure} of Algorithm \ref{Basis}. In order to generate the required state, the measurement in Step \ref{Basis_measure} needs to be performed for $\O(1/P_l)$ times.

Suppose $\lambda_1^2 \geq \lambda_2^2 \geq \cdots \geq \lambda_r^2$, where $\lambda_i$ is the eigenvalue of $\bm{H}$. Since state $|\bm{t}_m\>$ is the linear sum of $\{|\bm{h}_j\>\}_{j=1}^{d}$, we can assume that $|\bm{t}_m\>$ has the decomposition $|\bm{t}_m\> = \sum_{i=1}^r w_{mi}|\bm{u}_i\>$, for all $m=1,2,\cdots, l$, where $\sum_{i=1}^{r} w_{mi}w_{ni} = \delta_{mn}$. There is:
\begin{align*}
P_l &=  	\frac{1}{\|\bm{H}\|_F^2} \sum_{j=1}^{d} \left[ \|\bm{h}_j\|^2 \| |\bm{h}_j\> - \sum_{m=1}^l |\bm{t}_m\>\<\bm{t}_m|\bm{h}_j\> \|^2 \right]\\
&= \frac{1}{\|\bm{H}\|_F^2} \sum_{j=1}^d \left[ \|\bm{h}_j\|^2 -\sum_{m=1}^l \|\bm{h}_j\|^2 |\<\bm{t}_m|\bm{h}_j\>|^2 \right]\\
&= 1-\frac{1}{\|\bm{H}\|_F^2} \sum_{j=1}^d \sum_{m=1}^l \left[ \sum_{i=1}^{r} w_{mi} \lambda_i u_i^{(j)} \right]^2\\
&= 1-\frac{1}{\|\bm{H}\|_F^2} \sum_{j=1}^d \sum_{m=1}^l \left[ \sum_{i=1}^{r} w_{mi}^2 \lambda_i^2 ({u_i^{(j)}})^2 +\sum_{i \neq k}^{r} w_{mi}w_{mk} \lambda_i \lambda_k u_i^{(j)}u_k^{(j)} \right]\\
&= 1-\frac{1}{\|\bm{H}\|_F^2} \sum_{m=1}^l \sum_{i=1}^r w_{mi}^2 \lambda_i^2 \\
&= 1-\frac{1}{\|\bm{H}\|_F^2} \sum_{i=1}^{r} c_i \lambda_i^2,
\end{align*}
where $c_i=\sum_{m=1}^{l} w_{mi}^2$.

Consider the $r$-dimensional vector $\bm{w}_m=\sum_{i=1}^r w_{mi}\bm{e}_i$. The vector set $\{\bm{w}_m\}_{m=1}^{l}$ forms an orthogonal basis in a $l$-dimensional subspace. Note that we can add $\bm{w}_{l+1},\cdots \bm{w}_{r}$ such that $\{\bm{w}_m\}_{m=1}^{r}$ forms an orthonormal basis in the whole $r$-dimensional space. Denote matrix $W=(\bm{w}_1,\bm{w}_2,\cdots,\bm{w}_r)$. Since $W^{T}W=I$. Since $W$ is unitary, there is:
\begin{equation}
\sum_{m=1}^r w_{mi}^{2}=1, \forall i \in [r].	
\end{equation}
Thus we have the upper bound: $c_i = \sum_{m=1}^{l} w_{mi}^2 \leq \sum_{m=1}^r w_{mi}^{2}=1$. Note that $\sum_{i=1}^{r} c_i = \sum_{i=1}^r \sum_{m=1}^{l} w_{mi}^2 = \sum_{m=1}^{l} \sum_{i=1}^r w_{mi}^2 = l$, so there is:
\begin{equation}\label{Basis_P}
P_l \geq 1-\frac{1}{\|\bm{H}\|_F^2}\sum_{i=1}^{l} \lambda_i^2 = \frac{\sum_{i=l+1}^{r} \lambda_i^2}{\|\bm{H}\|_F^2}.
\end{equation}

Note that for the case $|\lambda_i| \leq \epsilon/2$, $\tilde{\lambda}_i=0$ is a good estimation for the \textbf{NCF problem}, so we could further assume $|\lambda_i|>\epsilon/2$ for the general case and bound the inequality (\ref{Basis_P}) as $P_l \geq \frac{(r-l)\epsilon^2}{4\|\bm{H}\|_F^2}$. 

Denote $T_{basis}$ as the required time to implement Algorithm \ref{Basis} and $T_{R_i}$ as the required time to implement operation $R_i$. Since in each iteration of $l \in [r-1]$, Algorithm \ref{Basis} refers operation $R_1,R_2,\cdots,R_l$ for $1/P_l$ times, there is:
\begin{align*}
\O(T_{basis}) &=\O(T_H r^8 \|\bm{H}\|_F^4\epsilon^{-4})+\sum_{l=0}^{r-1} \frac{1}{P_l} \sum_{m=1}^{l}\O(T_{R_i}) \\
&\leq \O(T_H r^8 \|\bm{H}\|_F^4\epsilon^{-4})+ \frac{4\|\bm{H}\|_F^2}{\epsilon^2} \O(T_Hr^{2\log(4r\|\bm{H}\|_F/\epsilon)}) \sum_{l=0}^{r-1}\frac{l}{r-l} \\
&=\O(T_H\|\bm{H}\|_F^2 \epsilon^{-2} (r^8 \|\bm{H}\|_F^2 \epsilon^{-2}+ r^{1+2\log(4r\|\bm{H}\|_F/\epsilon)}))\\
&\leq \O(T_H \poly(r)\epsilon^{-2}r^{2\log(4r\|\bm{H}\|_F/\epsilon)}).
\end{align*}

\begin{theorem}\label{Basis_theorem}
The Algorithm \ref{Basis} takes time $\O(T_H \poly(r)\epsilon^{-2}r^{2\log(4r\|\bm{H}\|_F/\epsilon)})$ to find an index set $\{g(i)\}_{i=1}^{r}$, which forms a complete basis $\{|\bm{h}_{g(i)}\>\}_{i=1}^{r}$ with probability at least 3/4.	
\end{theorem}

We also provide Lemma \ref{4.2_check} which gives the time complexity of confirming whether a given set $\{\bm{s}_i\}_{i=1}^{r}$ is linear independent or not. The proof is in Appendix A.
\begin{lemma}\label{4.2_check}
It takes $\O(r^3)$ time to check whether the set $\{\bm{s}_i\}_{i=1}^{r}$ is linear independent when the classical access to Hessian $\bm{H}$ is given, where $\bm{s}_i$ is sampled from column vectors of matrix $\bm{H}$.
\end{lemma}

\subsection{Coordinates Estimation}
\label{4_coordinate_estimate}

Assume the complete basis \(\{ |\bm{s}_1 \>,\  |\bm{s}_2 \>,\  \cdots, |\bm{s}_r \> \}\) has been selected out in Section \ref{4_basis}. 
Thus the read-out problem could be viewed as solving the equation \(|\bm{u}_t\> = \sum_{i=1}^r x_i |\bm{s}_i\>\), where \(x_i \in \R\) are unknown variables. We propose Algorithm \ref{coordinate} to calculate the coordinate $x_j$.
The main idea is to solve the \(r\)-dimensional linear equation \(\C\bm{x}=\bm{b}\), where \(b_i = \< \bm{u}_t | \bm{s}_i \>\) and \(c_{ij} = \< \bm{s}_i | \bm{s}_j \>\) for \(i,j \in [r]\). This equation could be solved classically in at most \(\O(r^3)\) time. 
Note that we can only get the approximation to \(c_{ij}\) or \(b_i\) instead of the exact value. Theorem \ref{84y564} verifies the impact of the approximate error to $c_{ij}$ or $b_i$ on the read-out of the target state. The proof of Theorem \ref{84y564} is in the Appendix A. 

  \begin{algorithm}[t]
  \caption{Coordinate Estimation}
  \label{coordinate}
  \begin{algorithmic}[1]
    \Require
     Quantum access to oracle $U_H$. The complete basis $\{|\bm{s}_i\>\}_{i=1}^{r}$. The target vector $|\bm{u}_t\>$.
    \Ensure
    Estimation $\tilde{\bm{x}}$ to the coordinate \(|\bm{u}_t\> = \sum_{i=1}^r x_i |\bm{s}_i\>\).
     \For{$i=1$ to $r$}
     \State Estimate the overlap $\<\bm{u}_t|\bm{s}_i\>$ and store the value in $\tilde{b}_{i}$.
     \For{$j=1$ to $r$}
     \State Estimate the overlap $\<\bm{s}_i|\bm{s}_j\>$ and store the value in $\tilde{c}_{ij}$.
     \EndFor
     \EndFor
    \State Create the vector $\bm{\tilde{b}} \in \R^{r}$ whose $i$-th component is $\tilde{b}_{i}$. Create the matrix $\tilde{\bm{C}} \in \R^{r \times r}$ whose $ij$-th component is $\tilde{c}_{ij}$.
    \State Solve the linear system $\tilde{\bm{C}}\bm{x}=\bm{\tilde{b}}$ and output the solution $\tilde{\bm{x}}$.
  \end{algorithmic}
\end{algorithm}

\begin{theorem}
\label{84y564}
Suppose \(\tilde{c}_{jk}\) is the \(\epsilon_1\)-approximation to \(c_{jk}=\< \bm{s}_j | \bm{s}_k \>\) and \(\tilde{b}_j\) is the \(\epsilon_2\)-approximation to \(b_j=\< \bm{u}_t | \bm{s}_j \>\) \(\forall j,k\in[r]\), where $\epsilon_1 = \frac{\epsilon}{6r^2 \|\C^{-1}\|^2}$ and $\epsilon_2=\frac{\epsilon}{6r\|\C^{-1}\|}$. Denote vector \(\tilde{\bm{x}} \in \R^r\) as the solution of \(\tilde{\C}\bm{x}=\tilde{\bm{b}}\). Then \(\tilde{\bm{x}}\) could lead an approximate eigenvector \(\tilde{\bm{u}}_t=\sum_{j=1}^{r}\tilde{x}_j \bm{s}_j\), such that $\|\boldsymbol{\tilde{u}}_t-\boldsymbol{u}_t\|\leq \epsilon/2$.
\end{theorem}

We propose several quantum algorithms in Appendix B to estimate overlap \(b_i = \< \bm{u}_t | \bm{s}_i \>\) and \(c_{ij} = \< \bm{s}_i | \bm{s}_j \>\), which are based on the Quantum SWAP Test\cite{buhrman2001quantum}. 
Our proposed quantum algorithms could present $\epsilon_1$-estimation to \(c_{ij} = \< \bm{s}_i | \bm{s}_j \>\) in time $\O(T_H\epsilon_1^{-2})$ and $\epsilon_2$-estimation to \(b_{i} = \< \bm{u}_t | \bm{s}_i \>\) in time $\O((T_{Input}+T_H)(\epsilon_2^{-4}+\epsilon_2^{-2}\|\bm{H}\|_F^2))$, where $T_{Input}$ is the time to generate state $|\bm{u}_t\>$. 
Since the time complexity to generate target state is $\O(T_H\|\bm{H}\|_F^3 \polylog(d) \epsilon^{-1})$ as proposed in Theorem~\ref{state_finding}, we could derive Corollary~\ref{solution}.

\begin{corollary}
\label{solution}
The classical description of the target state ${|\bm{u}}_t\>=\sum_{i=1}^{r} {x}_i |\bm{s}_i\>$ could be presented in time \(\O(T_H \polylog(d)\poly(r) \epsilon^{-5})\) with error bounds in \(\epsilon/2\), when the complete basis set $\{\bm{s}_j\}_{j=1}^{r}$ is given.
\end{corollary}

Considering the time complexity $\O(T_H\|\bm{H}\|_F^5\polylog(d)\epsilon^{-1})$ to label the proper eigenvalue and the time complexity $\O(T_H \poly(r)\epsilon^{-2}r^{2\log(4r\|\bm{H}\|_F/\epsilon)})$ to generate the complete basis set, we could solve the Negative Curvature Finding problem in time \(\O(T_H\polylog(d)\poly(r)\epsilon^{-2}(\epsilon^{-3}+r^{2\log(4r\|\bm{H}\|_F/\epsilon)}))\) by providing the target vector in the form $\bm{u}_t=\sum_{i=1}^r x_i \bm{h}_{g(i)}/\|\bm{h}_{g(i)}\|$ with error bounds in $\epsilon$ or making the none-vector statement.

\section{Conclusion}
\label{QNCD_conclu}
We propose an efficient quantum model for the Negative Curvature Finding problem, which is important for many second-order methods in non-convex optimization. The proposed quantum algorithm could produce the target state in time \(\O(T_H\epsilon^{-1}\poly(r)\polylog(d))\) with probability \(1-1/\poly(d)\), which runs exponentially faster than existing classical methods. Moreover, we propose an efficient hybrid quantum-classical algorithm for the efficient classical read-out of the target state with time complexity \(\O(T_H\poly(r)\polylog(d)\epsilon^{-2}(\epsilon^{-3}+r^{2\log(4r\|\bm{H}\|_F/\epsilon)}))\), which is exponentially faster on the degree of $d$ than existing general quantum state read-out methods.

\bibliographystyle{unsrt}
\bibliography{NIPS2019}

\begin{appendices}

\section{}

\textbf{The proof of Lemma \ref{bounded_norm}:}

\begin{proof}
Assume \(\lambda_1 \leq \lambda_2 \leq \cdots \leq \lambda_d\) are eigenvalues of $\bm{H}$, we have:

\begin{equation*}
\min_{\|\bm{v}\|=1} \bm{v}^{T}\bm{H}\bm{v} \leq \lambda_j \leq \max_{\|\bm{v}\|=1} \bm{v}^{T}\bm{Hv},\ \forall j \in [d].
\end{equation*}

By the definition of the Hessian matrix, for unit vector \(\bm{v}\), we have:
\begin{equation*}
\bm{Hv} =\nabla^2 f(\bm{x}) \bm{v} = \lim_{{h} \rightarrow 0} \frac{\nabla f(\bm{x}+h\bm{v})- \nabla f(\bm{x})}{{h}}.
\end{equation*}
From above equation, we can obtain: 

\quad\quad \(\bm{v}^{T}\bm{Hv} \leq \|\bm{v}\|\cdot \|\bm{Hv}\| \leq \lim_{\bm{h} \rightarrow 0} \frac{\|\nabla f(\bm{x}+h\bm{v})- \nabla f(\bm{x})\|}{h} \leq \lim_{{h} \rightarrow 0} \frac{L \|h\bm{v}\|}{h} = L\),

and  \(\bm{v}^{T}\bm{Hv} \geq -\|\bm{v}\|\cdot \|\bm{Hv}\| \geq -\lim_{h \rightarrow 0} \frac{\|\nabla f(\bm{x}+h\bm{v})- \nabla f(\bm{x})\|}{h} \geq -\lim_{h \rightarrow 0} \frac{L \|h\bm{v}\|}{h} = -L\).

Thus, the eigenvalue \(\lambda_j\) is bounded in \([-L,L]\) for all \(j \in [d]\).

We have 
\(\|\bm{H}\|_F^2=\sum_{i} \sum_{j} h_{ij}^2 = Tr(\bm{H}\cdot \bm{H}) = \sum_j{\lambda_j (\bm{H}^2)} \leq r L^2\) , so \(\|\bm{H}\|_F \leq \sqrt{r} L\).
\end{proof}

\begin{lemma}
{\rm \textbf{Hoeffding's inequality\cite{hoeffding1994probability}}}

Suppose \(X_1, X_2, \cdots, X_n\) are independent random variables with bounds \(X_i \in [a_i, b_i], \forall i \in [n]\). Define \(\overline{X}=\frac{1}{n}\sum_{i=1}^{n}X_i\) , then \(\forall \epsilon >0\), we have:
\begin{equation}
P(\overline{X}-E[\overline{X}] \ge \epsilon) \le \exp{(-\frac{2n^2 \epsilon^2}{\sum_{i=1}^{n} (b_i - a_i)^2})},
\end{equation}
and
\begin{equation}
P(\overline{X}-E[\overline{X}] \le -\epsilon) \le \exp{(-\frac{2n^2 \epsilon^2}{\sum_{i=1}^{n} (b_i - a_i)^2})}.
\end{equation}

\end{lemma}

\textbf{The proof of Theorem \ref{positive_negative}:}
\begin{proof}

The measurement in step \ref{Qcurvature_negaeigen_mea} of Algorithm \ref{negative_eigenvalue} outputs \(1\) with probability:
\begin{equation}
P(1) = \|\frac{|\bm{Pu}\> -|\bm{Qu}\>}{2}\|^2 = \frac{1}{4} ( \langle \bm{Pu}| - \langle \bm{Qu}|)(|\bm{Pu}\rangle -|\bm{Qu}\rangle). 
\end{equation}

Note that \(\langle \bm{Pu}|\bm{Qu} \rangle = \bm{u}^{T} \bm{P}^{T} \bm{Qu} = \frac{1}{\|\bm{H}\|_F}\bm{u}^{T} \bm{H}\bm{u} = \frac{\lambda}{\|\bm{H}\|_F}\), so \(\bm{P}(1)=\frac{1-\lambda/\|\bm{H}\|_F}{2}\) . Similarly we have \(\bm{P}(0)=\frac{1+\lambda/\|\bm{H}\|_F}{2}\).

Suppose that we need \(2x+1\) times of measurement to give an \(1-\delta\) correct statement about whether \(\lambda >0\)  or \(\lambda <0\). The problem can be viewed as the biased coin problem. Define random variables \(X_i\) such that \(\bm{P}(X_i=1)=p\) and \(\bm{P}(X_i=0)=1-p\) and \(S_n = \sum_{i=1}^{n} X_i\). Then there has the Hoeffding's inequality \(\bm{P} ( S_{n} / n-p \leq -\epsilon) \leq e^{-2 n \epsilon^{2}}\) and  \(\bm{P} ( S_{n} / n-p \geq \epsilon) \leq e^{-2 n \epsilon^{2}}\).

Back to the problem, suppose $\lambda<0$, by setting \(n=2x+1,\ n(p-\epsilon)=x\) and \(\bm{P}=\frac{1-\lambda/\|\bm{H}\|_F}{2}\), we have:
\begin{equation*}
P(S_{2x+1} \leq x) \leq \exp(-2(2x+1)[\frac{1-\lambda/\|\bm{H}\|_F}{2}-\frac{x}{2x+1}]^2) < e^{-\frac{2x+1}{2}\frac{\lambda^2}{\|\bm{H}\|_F^2}} \leq e^{-\frac{2x+1}{2}\frac{a^2}{\|\bm{H}\|_F^2}}. 
\end{equation*}

Similarly for $\lambda > 0$, there is $P(S_{2x+1} \geq x) \leq e^{-\frac{2x+1}{2}\frac{a^2}{\|\bm{H}\|_F^2}}$. 

Let \(e^{-\frac{2x+1}{2}\frac{a^2}{\|\bm{H}\|_F^2}} \leq\delta\)  , we have \(x \geq [\frac{\|\bm{H}\|_F^2}{a^2}\log{\frac{1}{\delta}}-\frac{1}{2}]+1\).
\end{proof}

\textbf{The proof of Lemma 2}:
\begin{proof}
Define the index function $g:[r]\rightarrow [d]$ such that $\bm{s}_i=\bm{h}_{g(i)}, \forall i \in [r]$.
Consider the eigen-decomposition of matrix $\bm{H}$:
\begin{equation}
\bm{H} = \sum_{j=1}^{r} \lambda_j \bm{u}_j \bm{u}_j^{T}.	
\end{equation}
It is natural to generate the decomposition:
\begin{equation}
\bm{h}_j = \sum_{i=1}^{r} \lambda_i \bm{u}_i u_i^{(j)},	
\end{equation}
\begin{equation}
h_{jk} = \sum_{i=1}^{r} \lambda_i u_i^{(j)} u_i^{(k)}.	
\end{equation}
Define the $r \times r$ dimensional matrix $\C=(\bm{h}_{g(1)}^{T},\bm{h}_{g(2)}^{T},\cdots, \bm{h}_{g(r)}^{T})^{T} (\bm{h}_{g(1)},\bm{h}_{g(2)},\cdots, \bm{h}_{g(r)})$. There is:
\begin{equation}
\{\bm{h}_{g(i)}\}_{i=1}^{r}\ is\ linear\ independent \Leftrightarrow det(\C) \neq 0.	
\end{equation}
Denote the $jk$-th element of $\C$ as $c_{jk}$. Since $c_{jk}=\bm{h}_j^{T} \bm{h}_k=\sum_{i=1}^{r} \lambda_i^2 u_i^{(j)}u_i^{(k)}$, there is:
\begin{align}
det(\C)&={
\left| \begin{array}{ccc}
\sum_{i=1}^{r}\lambda_i^2 u_i^{(g(1))}u_i^{(g(1))} & \cdots & \sum_{i=1}^{r}\lambda_i^2 u_i^{(g(1))}u_i^{(g(r))}\\
\vdots & \ddots & \vdots\\
\sum_{i=1}^{r}\lambda_i^2 u_i^{(g(r))}u_i^{(g(1))} & \cdots & \sum_{i=1}^{r}\lambda_i^2 u_i^{(g(r))}u_i^{(g(r))}
\end{array} 
\right|}\\
&=\sum_{i_1 = 1}^{r}\sum_{i_2 = 1}^{r}\cdots\sum_{i_r = 1}^{r}{
\left| \begin{array}{ccc}
\lambda_{i_1}^2 u_{i_1}^{(g(1))}u_{i_1}^{(g(1))} & \cdots & \lambda_{i_r}^2 u_{i_r}^{(g(1))}u_{i_r}^{(g(r))}\\
\vdots & \ddots & \vdots\\
\lambda_{i_1}^2 u_{i_1}^{(g(r))}u_{i_1}^{(g(1))} & \cdots & \lambda_{i_r}^2 u_{i_r}^{(g(r))}u_{i_r}^{(g(r))}
\end{array} 
\right|}\\
&=\sum_{i_1 = 1}^{r}\sum_{i_2 = 1}^{r}\cdots\sum_{i_r = 1}^{r}
(\prod_{j=1}^{r} \lambda_{i_j}^2)(\prod_{j=1}^{r} u_{i_j}^{(g(j))}){\left| \begin{array}{ccc}
u_{i_1}^{(g(1))} & \cdots & u_{i_r}^{(g(1))}\\
\vdots & \ddots & \vdots\\
u_{i_1}^{(g(r))} & \cdots & u_{i_r}^{(g(r))}
\end{array} 
\right|}.\label{independent_C}
\end{align}
On the other hand, construct the matrix $\bm{H}'$ whose $jk$-th element is ${h}_{jk}' = h_{g(j),g(k)}$. There is:
\begin{align}
det(\bm{H}')&={
\left| \begin{array}{ccc}
\sum_{i=1}^{r}\lambda_i u_i^{(g(1))}u_i^{(g(1))} & \cdots & \sum_{i=1}^{r}\lambda_i u_i^{(g(1))}u_i^{(g(r))}\\
\vdots & \ddots & \vdots\\
\sum_{i=1}^{r}\lambda_i u_i^{(g(r))}u_i^{(g(1))} & \cdots & \sum_{i=1}^{r}\lambda_i u_i^{(g(r))}u_i^{(g(r))}
\end{array} 
\right|}\\
&=\sum_{i_1 = 1}^{r}\sum_{i_2 = 1}^{r}\cdots\sum_{i_r = 1}^{r}{
\left| \begin{array}{ccc}
\lambda_{i_1} u_{i_1}^{(g(1))}u_{i_1}^{(g(1))} & \cdots & \lambda_{i_r} u_{i_r}^{(g(1))}u_{i_r}^{(g(r))}\\
\vdots & \ddots & \vdots\\
\lambda_{i_1} u_{i_1}^{(g(r))}u_{i_1}^{(g(1))} & \cdots & \lambda_{i_r} u_{i_r}^{(g(r))}u_{i_r}^{(g(r))}
\end{array} 
\right|}\\
&=\sum_{i_1 = 1}^{r}\sum_{i_2 = 1}^{r}\cdots\sum_{i_r = 1}^{r}
(\prod_{j=1}^{r} \lambda_{i_j})(\prod_{j=1}^{r} u_{i_j}^{(g(j))}){\left| \begin{array}{ccc}
u_{i_1}^{(g(1))} & \cdots & u_{i_r}^{(g(1))}\\
\vdots & \ddots & \vdots\\
u_{i_1}^{(g(r))} & \cdots & u_{i_r}^{(g(r))}
\end{array} 
\right|}.\label{independent_H'}
\end{align}
Note that the determinant in eq(\ref{independent_C}) and eq(\ref{independent_H'}) is non-zero only if $i_m \neq i_n$ for any different $m,n \in [r]$. Consider the summation of $i_j$ for all $j \in [r]$ over $\{1,2,\cdots,r\}$, there is:
\begin{equation}
det(\C)/\prod_{i=1}^r \lambda_i^2 = det(\bm{H}')/\prod_{i=1}^r \lambda_i	
\end{equation}
Thus the problem about whether group $\{\bm{h}_{g(i)}\}_{i=1}^{r}$ is linear independent could be solved by calculating the determinant of matrix $\bm{H}'$. Since $\bm{H}'$ is a $r \times r$ dimensional matrix, $det(\bm{H}')$ could be calculated in $\O(r^3)$ time\cite{schwarzenberg1995matrix}. We could claim that the group $\{\bm{h}_{g(i)}\}_{i=1}^{r}$ is linear independent if $det(\bm{H}')\neq 0$, or $\{\bm{h}_{g(i)}\}_{i=1}^{r}$ is linear dependent if $det(\bm{H}')=0$.
\end{proof}

\textbf{The proof of Theorem \ref{84y564}:}
\begin{proof}
For  $|\Delta c_{ij}| \leq \epsilon_1$ and $|\Delta b_j|\leq \epsilon_2$, there is:
\begin{equation*}
\|\Delta \C\|\leq r \epsilon_1 \ \text{and}\ \|\Delta \b\| \leq \sqrt{r}\epsilon_2.
\end{equation*}
The matrix norm $\|\cdot\|$ here denotes the largest eigenvalue of the matrix.
Note that elements of matrix $\C$ and vector $\b$ are overlap of states, which are bounded in $[-1,1]$, so similarly there is: 
\begin{equation*}
\|\C\|\leq r \ \text{and}\ \| \b\| \leq \sqrt{r}.
\end{equation*}
There is:
\begin{align*}
\|\Delta\bm{x}\|&= \|(\C+\Delta \C)^{-1}(\Delta\b-\Delta \C\cdot \C^{-1}\b)\|\\
&\leq \|\C^{-1}\| \cdot \| (I+\C^{-1}\D\C)^{-1}\|\cdot (\|\D\b\|+\|\D\C \cdot \C^{-1}\b\|) \\
& \leq \|\C^{-1}\|\cdot\frac{1}{1-\|\C^{-1}\D\C\|}\cdot(\|\D\b\|+\|\D\C\|\|\C^{-1}\|\|\b\|)\\
& \leq \frac{\|\C^{-1}\|}{1-\|\C^{-1}\| r \epsilon_1}\cdot(\sqrt{r}\epsilon_2+r^{3/2} \epsilon_1\|\C^{-1}\|) \leq \frac{\epsilon}{2\sqrt{r}}.
\end{align*}

Thus, for ${\bm{u}}_t=\sum_{j=1}^{r}{x}_j \bm{s}_j$ and $\tilde{\bm{u}}_t=\sum_{j=1}^{r}\tilde{x}_j \bm{s}_j$, there is:
\begin{equation*}
\|\bm{u}_t-\tilde{\bm{u}}_t\|=\sqrt{\D\bm{x}^{T}\C\D\bm{x}}\leq \|\D\bm{x}\|\cdot\|\C\|^{1/2}\leq \frac{\epsilon}{2\sqrt{r}} \cdot \sqrt{r}=\frac{\epsilon}{2}.
\end{equation*}
\end{proof}

\section{}

\subsection{The estimation of $c_{ij}=\<{s}_i|{s}_j\>$:}

The overlap $c_{ij}=\<\bm{s}_i|\bm{s}_j\>=\<\bm{h}_{g(i)}|\bm{h}_{g(j)}\>$ can be estimated by the Hadamard Test. We provide the detail in Algorithm \ref{c_ij}:

\begin{algorithm}[htb]
  \caption{$c_{ij}$ estimation}
  \label{c_ij}
  \begin{algorithmic}[1]
    \Require
     Quantum access to oracle $U_H$. The index number $i$ and $j$. The precision parameter $\epsilon$. The probability error bound $\delta$.
    \Ensure
    An estimation $\tilde{c}_{ij}$ to the value $c_{ij}=\<\bm{s}_i|\bm{s}_j\>$, such that $\tilde{c}_{ij} \in c_{ij} \pm \epsilon$ with probability at least $1-\delta$.
     \For{$k=1$ to $n=[\frac{2}{\epsilon^2}\log(\frac{2}{\delta})]+1$}
     \State Create state $[|\bm{s}_i\>|0\>+|\bm{s}_j\>|1\>]/\sqrt{2}$.
     \label{c_ij_init}
     \State Apply the Hadmard gate on the second register to obtain the state $\frac{|\bm{s}_i\>+|\bm{s}_j\>}{2}|0\>+\frac{|\bm{s}_i\>-|\bm{s}_j\>}{2}|1\>$.\label{c_ij_hadmard}
     \State Measure the second register and record the result.
     \label{c_ij_measure}
     \EndFor
     \State Count the number of resulting $0$ in step \ref{c_ij_measure} as $m$. Output $2m/n-1$ as the estimation to $c_{ij}$.
     \label{c_ij_output}
  \end{algorithmic}
\end{algorithm}

\begin{theorem}
\label{c_ij_theorem}
Algorithm~\ref{c_ij} present the $\epsilon$-estimation to the overlap $c_{ij}=\<\bm{s}_i|\bm{s}_j\>$ with probability at least $1-\delta$ with running time $\O(T_H\polylog(d)\epsilon^{-2}\log(1/\delta))$.
\end{theorem}

We generate state $\frac{|\bm{s}_i\>|0\>+|\bm{s}_j\>|1\>}{\sqrt{2}}=\frac{|\bm{h}_{g(i)}\>|0\>+|\bm{h}_{g(j)}\>|1\>}{\sqrt{2}}$  in step \ref{c_ij_init} by performing the following procedure on state $|g(i)\>|g(j)\>|0\>|0\>$:
\begin{align}
	|g(i)\>|g(j)\>|0\>|0\> 
	\stackrel{H}{\longrightarrow} 
	&|g(i)\>|g(j)\>|0\>\frac{|0\>+|1\>}{\sqrt{2}}\label{b_i_hadmard}\\
	\xrightarrow{U_H \otimes |0\>\<0|}
	&|g(i)\>|g(j)\>\frac{|\bm{h}_{g(i)}\>|0\>+|0\>|1\>}{\sqrt{2}}\label{b_i_gi}\\
	\xrightarrow{U_H \otimes |1\>\<1|}
	&|g(i)\>|g(j)\>\frac{|\bm{h}_{g(i)}\>|0\>+|\bm{h}_{g(j)}\>|1\>}{\sqrt{2}}\label{b_i_gj}\\
	\xrightarrow{\text{trace out the first two registers}}
	&\frac{|\bm{h}_{g(i)}\>|0\>+|\bm{h}_{g(j)}\>|1\>}{\sqrt{2}}.\label{b_i_trace}
\end{align}

The Hadmard gate in (\ref{b_i_hadmard}) acts on the 4-th register. The gate $U_H \otimes |0\>\<0|$ in (\ref{b_i_gi}) acts on the 1-st, 3-rd and 4-th registers. The gate $U_H \otimes |1\>\<1|$ in (\ref{b_i_gj}) acts on the 2-nd, 3-rd and 4-th registers.

\subsection{The estimation of $b_{i}=\<u_t|{s}_i\>$:}
The estimation to the overlap $b_i=\<\bm{u}_t|\bm{s}_i\>$ is more complicated. Technics like Algorithm~\ref{c_ij} is infeasible, due to the post-selection method for generating target state $|\bm{u}_t\>$. Here we introduce another standard quantum algorithm named as Quantum Swap Test\cite{buhrman2001quantum}, which could estimate the square overlap between two quantum states $|\phi\>$ and $|\psi\>$. The circuit of the Quantum Swap Test is illustrated in Figure \ref{swap_test}.

\begin{figure}[htb]
  \centering
  \includegraphics[width=.5\textwidth]{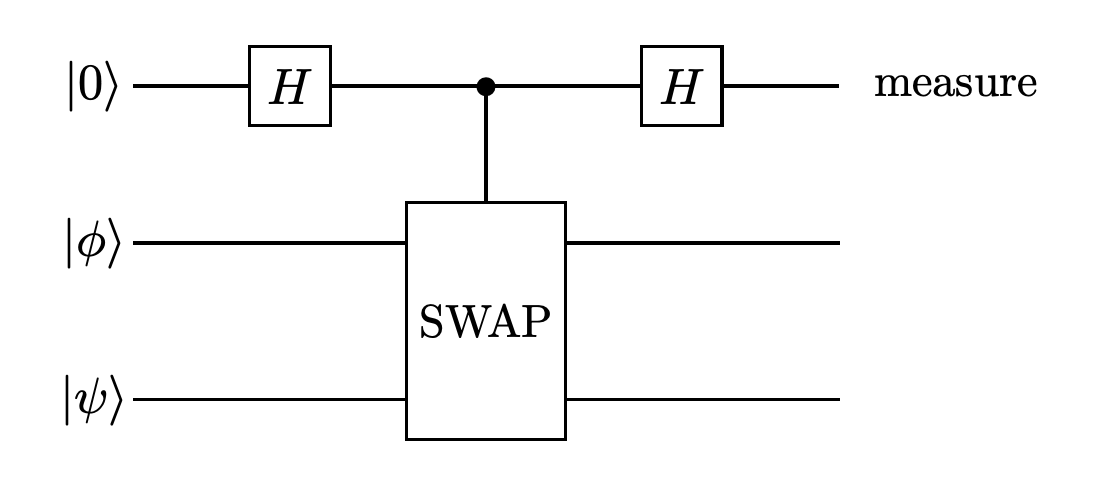}
  \caption{Circuit of the Quantum Swap Test}
  \label{swap_test}
\end{figure}

As shown in Figure \ref{swap_test}, Quantum Swap Test performs the operation:
\begin{equation}
|0\>|\phi\>|\psi\> \rightarrow (H \otimes I)(|0\>\<0|\otimes I + |1\>\<1|\otimes U_{SWAP})(H \otimes I)|0\>|\phi\>|\psi\>.
\end{equation}
The final state could be written as:
\begin{equation}
\frac{1}{2}|0\>(|\phi\>|\psi\>+|\psi\>|\phi\>)+\frac{1}{2}|1\>(|\phi\>|\psi\>-|\psi\>|\phi\>).
\end{equation}
The $U_{SWAP}$ gate could be implemented in time $\O(\polylog(d))$, which performs the swap transformation $|\phi\>|\psi\> \rightarrow |\psi\>|\phi\>$ for $d$-dimensional state $|\phi\>$ and $|\psi\>$. The measurement on the first qubit produces $0$ with probability $P_0 = \frac{1}{2}(1+|\<\phi|\psi\>|^2)$. Thus, by replacing step \ref{c_ij_init}-\ref{c_ij_hadmard} in Algortihm~\ref{c_ij} with the Quantum Swap Test operation, we could build an algorithm to estimate the square of the state overlap.

\begin{theorem}
\label{QST_theorem}
There exists a quantum algorithm which could present $\epsilon$-estimation to value $b_i^2=|\<\bm{u}_t|\bm{s}_i\>|^2$ with probability at least $1-\delta$ in running time $\O(T_{Input}\epsilon^{-2}\log(1/\delta))$, where $T_{Input}$ is the time complexity to generate states $|\bm{u}_t\>$ and $|\bm{s}_i\>$.
\end{theorem}

In order to estimate values $b_i=\<\bm{u}_t|\bm{s}_i\>$ for $i\in[r]$, we need to discriminate the positive and negative of $b_i$. 
Note that for state $|\bm{u}_t\>$, the state $|-\bm{u}_t\>$ is also a target state which shares the same eigenvalue. So both states $|\bm{u}_t\>$ and $|-\bm{u}_t\>$ are legal outputs and are indistinguishable for our algorithm in Section~\ref{QNCD_state_finding}. Thus we analysis the value $b_i=sgn({u}_t^{(k)})\<\bm{u}_t|\bm{s}_i\>$ as the overlap between states $|\bm{u}_t\>$ and $|\bm{s}_i\>$, where $\bm{u}_t^{(k)}$ is the $k$-th component of vector $\bm{u}_t$. Generally $k$ could be any index such that the corresponding component is none-zero. Here we choose the index such that the square overlap $|\<\bm{u}_t|\bm{h}_{i}\>|^2$ is the largest for all $i \in [d]$:
\begin{equation}
k=\mathop{\arg\max}_{i} |\<\bm{u}_t|\bm{h}_{i}\>|^2.
\end{equation}

The $\max_{i \in [r]}|\<\bm{u}_t|\bm{h}_{i}\>|^2$ has the lower bound:
\begin{equation}
\max_{i \in [d]}|\<\bm{u}_t|\bm{h}_{i}\>|^2=
\max_{i \in [d]} \frac{(\bm{u}_t^{T}\bm{h}_{i})^2}{\|\bm{h}_{i}\|^2} =\max_{i \in [d]} \frac{\lambda_t^2 {\bm{u}_t^{(i)}}^2}{\|\bm{h}_{i}\|^2} \geq \frac{\sum_{i=1}^d \lambda_t^2 {\bm{u}_t^{(i)}}^2}{\sum_{i=1}^d \|\bm{h}_i\|^2}=\frac{\lambda_t^2}{\|\bm{H}\|_F^2}.
\end{equation}

Note that $b_i=sgn({u}_t^{(k)})\<\bm{u}_t|\bm{s}_i\>=-sgn(\<\bm{u}_t|\bm{h}_k\>)\<\bm{u}_t|\bm{s}_i\>$. Define two states $|\psi_+\>=\frac{1}{Z_{+}}(|\bm{h}_k\>+|\bm{h}_{g(i)}\>)$ and $|\psi_-\>=\frac{1}{Z_{-}}(|\bm{h}_k\>-|\bm{h}_{g(i)}\>)$, where $Z_{\pm}$ are normalized constants such that $Z_{\pm}^2=2 \pm 2\<\bm{h}_k|\bm{h}_{g(i)}\>$. Then there is:
\begin{align}
& |\<\bm{u}_t|\psi_+\>|^2=\frac{1}{Z_+^2}\Big[\<\bm{u}_t|\bm{h}_k\>^2 + \<\bm{u}_t|\bm{h}_{g(i)}\>^2 + 2\<\bm{u}_t|\bm{h}_k\>\<\bm{u}_t|\bm{h}_{g(i)}\> \Big],\\
& |\<\bm{u}_t|\psi_-\>|^2=\frac{1}{Z_-^2}\Big[\<\bm{u}_t|\bm{h}_k\>^2 + \<\bm{u}_t|\bm{h}_{g(i)}\>^2 - 2\<\bm{u}_t|\bm{h}_k\>\<\bm{u}_t|\bm{h}_{g(i)}\> \Big].
\end{align}

States $|\psi_+\>$ and $|\psi_-\>$ could be generated by step~\ref{c_ij_init}-\ref{c_ij_measure} in Algorithm \ref{c_ij}. 
The overlap $\<\bm{h}_k|\bm{h}_{g(i)}\>$ could be estimated by Algorithm \ref{c_ij}. The square overlap $|\<\bm{u}_t|\psi_+\>|^2$ and $|\<\bm{u}_t|\psi_-\>|^2$ could be estimated by Quantum Swap Test. 
Thus for $|\<\bm{u}_t|\bm{h}_{g(i)}\>|>\epsilon'$, one could discriminate the positive and negative of $\<\bm{u}_t|\bm{h}_k\>\<\bm{u}_t|\bm{h}_{g(i)}\>$ by calculate the value $Z_+^2 |\<\bm{u}_t|\psi_+\>|^2-Z_-^2 |\<\bm{u}_t|\psi_-\>|^2$. 
The estimation on the square overlap $|\<\bm{u}_t|\psi_+\>|^2$ and $|\<\bm{u}_t|\psi_-\>|^2$ need to have the precision $\epsilon'|\lambda_t|/(2\|\bm{H}\|_F)$, which takes time $\O(T_{Input}\|\bm{H}\|_F^2\epsilon'^{-2})$. For $|\<\bm{u}_t|\bm{h}_{g(i)}\>| \le \epsilon'$, $0$ is an $\epsilon'$ estimation to $\<\bm{u}_t|\bm{h}_{g(i)}\>$. Since an $\epsilon'$-estimation to $|b_i|$ could be achieved by an $\epsilon'^2$-estimation to $b_i^2$ which takes time $\O(T_{Input}\epsilon'^{-4})$, we could derive the time complexity of estimating $b_i$ in Theorem \ref{b_i_theorem}.

\begin{theorem}
\label{b_i_theorem}
There exists a quantum algorithm which could present $\epsilon$-estimation to value $b_i=sgn(u_t^{(k)})\<\bm{u}_t|\bm{s}_i\>$ with probability at least $1-\delta$ in running time $\O(T_{Input}\polylog(d)(\epsilon^{-4}+\epsilon^{-2}\|\bm{H}\|_F^2)\log(1/\delta))$, where $T_{Input}$ is the time complexity to generate states $|\bm{u}_t\>$ and $|\bm{s}_i\>$.
\end{theorem}

\end{appendices}

\end{document}